\documentclass[article,3p]{elsarticle}

\usepackage{setspace}
\usepackage{amsmath,amsfonts,amssymb,amsbsy,amscd}
\usepackage{mathrsfs,array}
\usepackage{verbatim}
\usepackage{lipsum}
\usepackage{tikz}
\usetikzlibrary{3d}
\usetikzlibrary {arrows.meta} 
\usepackage{upgreek}
\usetikzlibrary{intersections, arrows, automata, backgrounds, calendar, chains, matrix, mindmap, patterns, petri, shadows, shapes.geometric, shapes.misc, spy, trees, shapes}
\usepackage{amsthm}
\usepackage{subcaption}
\usepackage{epstopdf}
\usepackage{booktabs}
\usepackage{empheq}
\usepackage{tabularx}

\usepackage{xargs}
\usepackage{xcolor}
\usepackage{soul}
\usepackage{hyperref}
\hypersetup{
    colorlinks=true,
    linkcolor=blue,
    urlcolor=red,
    linktoc=all
}

\DeclareMathAlphabet{\mathscrbf}{OMS}{mdugm}{b}{n}

\usepackage[makeroom]{cancel}
\usepackage{amsbsy}

\usepackage[nameinlink,capitalise,noabbrev]{cleveref}

\newtheorem{thm}{Theorem}

\numberwithin{thm}{section} 
\newtheorem{rmk}[thm]{Remark}
\newproof{pf}{Proof}
\newproof{pot}{Proof of Theorem \ref{thm2}}

\def\B#1{\mbox{\boldmath{$#1$}}}

\newcommand{\divg}{{\rm div}}

\newcommand{\nn}{\nonumber}

\newcommand{\bu}{\mathbf{u}}

\newcommand{\bw}{\mathbf{w}}

\newcommand{\bv}{\mathbf{v}}
\newcommand{\Bn}{\mathbf{n}}
\newcommand{\bx}{\mathbf{x}}

\newcommand{\bg}{\B{g}}
\newcommand{\bJ}{\mathbf{J}}

\newcommand{\WW}{\mathcal{W}}

\def\be{\begin{equation}}
\def\ee{\end{equation}}
\def\ba{\begin{array}}
\def\ea{\end{array}}
\def\bea{\begin{eqnarray}}
\def\eea{\end{eqnarray}}
\def\beas{\begin{eqnarray*}}
\def\eeas{\end{eqnarray*}}
\newcommand{\bseq}{\begin{subequations}}
\newcommand{\eseq}{\end{subequations}}

\biboptions{numbers,sort&compress}

\crefname{lem}{Lemma}{Lemmas}
\crefname{thm}{Theorem}{Theorems}
\crefname{prop}{Proposition}{Propositions}
\crefname{corol}{Corollary}{Corollaries}
\crefname{rmk}{Remark}{Remarks}

\makeatletter
\newcommand{\myref}[1]{\cref{#1}\mynameref{#1}{\csname r@#1\endcsname}}
\newcommand{\Myref}[1]{\Cref{#1}\mynameref{#1}{\csname r@#1\endcsname}}

\def\mynameref#1#2{%
  \begingroup
    \edef\@mytxt{#2}%
    \edef\@mytst{\expandafter\@thirdoffive\@mytxt}%
    \ifx\@mytst\empty\else
    \space(\nameref{#1})\fi
  \endgroup
}
\makeatother

\definecolor{darkgreen}{rgb}{0, 0.7, 0}
\definecolor{dartmouthgreen}{rgb}{0.05, 0.5, 0.06}

\DeclareRobustCommand{\reva}[1]{%
  \ifmmode\text{\textcolor{darkgreen}{$#1$}}\else\textcolor{darkgreen}{#1}\fi
}

\definecolor{darkblue}{rgb}{0.0, 0.5, 1.0}
\DeclareRobustCommand{\revb}[1]{%
  \ifmmode\text{\textcolor{red}{$#1$}}\else\textcolor{red}{#1}\fi
}

\definecolor{amber}{rgb}{1.0, 0.49, 0.0}
\DeclareRobustCommand{\revc}[1]{%
  \ifmmode\text{\textcolor{dartmouthgreen}{$#1$}}\else\textcolor{dartmouthgreen}{#1}\fi
}

\DeclareRobustCommand{\revtbd}[1]{%
  \ifmmode\text{\hlmagenta{$#1$}}\else\hlmagenta{#1}\fi
}

\newcolumntype{R}[2]{%
    >{\adjustbox{angle=#1,lap=\width-(#2)}\bgroup}%
    l%
    <{\egroup}%
}
\makeatletter
\newcommand{\thickhline}{%
    \noalign {\ifnum 0=`}\fi \hrule height 1pt
    \futurelet \reserved@a \@xhline
}
\newcolumntype{"}{@{\hskip\tabcolsep\vrule width 1pt\hskip\tabcolsep}}
\makeatother

\usepackage{pifont}
\newcommand{\cmark}{\ding{51}}%
\newcommand{\xmark}{\ding{55}}%
\newcommand{\smark}{\ding{72}}%

\makeatletter
\def\ps@pprintTitle{%
  \let\@oddhead\@empty
  \let\@evenhead\@empty
  \let\@oddfoot\@empty
  \let\@evenfoot\@oddfoot
}
\makeatother

\journal{Journal of Computational Physics}

\begin{document}
\begin{frontmatter}

\title{{\Large The divergence-free velocity formulation of the consistent Navier-Stokes Cahn-Hilliard model with non-matching densities, divergence-conforming discretization, and benchmarks}}

\author{M.F.P. ten Eikelder\corref{cor1}}
\cortext[cor1]{Corresponding author}
\ead{marco.ten-eikelder@tu-darmstadt.de}
\author{D. Schillinger}
\ead{dominik.schillinger@tu-darmstadt.de}

\address{Institute for Mechanics, Computational Mechanics Group, Technische Universit\"{a}t Darmstadt, Franziska-Braun-Str. 7, 64287 Darmstadt, Germany}
\date{\today}

\date{}

\begin{abstract}
The prototypical diffuse-interface model that describes multi-component flows is the Navier-Stokes Cahn-Hilliard model (NSCH). Over the last decades many NSCH models have appeared that claim to describe the same physical phenomena, yet are distinct from one another. In a recent article [M.F.P. ten Eikelder, K.G. van der Zee, I. Akkerman, and D. Schillinger, Math. Mod. Meth. Appl. S. 33, pp 175-221, 2023.] we have established a unified framework of virtually all NSCH models. 
The framework reveals that there is only a single consistent NSCH model that naturally emanates from the underlying mixture theory.
In the current article we present, verify and validate this novel consistent NSCH model by means of numerical simulation. To this purpose we discretize a divergence-free velocity formulation of the NSCH model using divergence-conforming isogeometric spaces. We compare computations of our consistent model to results of existing models from literature.
The predictive capability of the numerical methodology is demonstrated via three-dimensional computations of a rising bubble and the contraction of a liquid filament that compare well with experimental data.
\end{abstract}


\end{frontmatter}


\section{Introduction}
The Navier-Stokes Cahn-Hilliard (NSCH) model is a diffuse-interface flow model that describes the evolution of viscous incompressible isothermal fluid mixtures. This model has proven to be powerful approach for the simulation of free-surface problems in which topological changes, surface tension, and density jumps play an important role. Typical examples of these problems include the deformation, coalescence or break-up of bubble or droplets.

The development of NSCH models started with matching density models proposed by Hohenberg and Halperin \cite{hohenberg1977theory} and Gurtin \cite{gurtinmodel}. Over the last decades, the design of Navier-Stokes Cahn-Hilliard models has been extended to the non-matching density case with important contributions by Lowengrub and Truskinovsky \cite{lowengrub1998quasi}, Boyer  \cite{boyer2002theoretical}, Ding et al. \cite{ding2007diffuse}, Abels et al. \cite{abels2012thermodynamically}, Shen et al. \cite{shen2013mass},  Aki et al. \cite{aki2014quasi} and Shokrpour Roudbari et al. \cite{shokrpour2018diffuse}. Each of these works aims to describe the same physics, yet the proposed models are different. In particular, three problems occur: (1) the systems of balance laws of the various models are distinct before constitutive choices have been applied, (2) the energy-dissipation laws of the models are conflicting, and (3) some of the models are inconsistent in the single fluid regime. As such, there is no consensus in the realm of NSCH models, and the connection between the existing NSCH models has been mostly unexplored.

In our recent article \cite{eikelder2023unified} we have revealed the sought-after connections by establishing a unified framework of NSCH models with non-matching densities\footnote{This framework encompasses mass transfer between constituents resulting in models of Navier-Stokes Cahn-Hilliard Allen-Cahn type. In the current work we assume the absence of constituent mass transfer, and thus deal with Navier-Stokes Cahn-Hilliard models.}. Given that each of the models aims to describe the same physics, it is only natural to anticipate a \textit{single NSCH model}, not a collection of models. Of course, variations appear in the selection of constitutive models, but the general (form of the) model is fixed. Indeed, the unified framework naturally leads to a single NSCH model. This model is invariant to the set of fundamental variables. After applying small but important corrections, and subsequently applying simple variable transformations, the connection to virtually all existing NSCH models is established. For example, the longstanding belief is that there are two distinct classes of models, (i) models with mass-averaged velocities and (ii) models with volume-averaged velocities. Our framework indicates that these seemingly different classes of models are in fact two sides of the same coin; simple variable transformations reveal that these classes of models are equivalent.

Since the NSCH models describe the flow equations of a mixture of fluids, the correct theoretical framework is \textit{continuum mixture theory}, as proposed by Truesdell and Toupin \cite{truesdell1960classical,truesdell1984historical}. It appears that existing NSCH models are only partly established via mixture theory. As a consequence, either the balance laws are not compatible with mixture theory before constitutive modeling, or the final model is inconsistent in the single fluid regime. As such, none of the NSCH models in literature is fully compatible with the NSCH model, see \cref{table: overview models constit mod}, and (small) rectifications are necessary to match the \textit{consistent} NSCH model that we have introduced in \cite{eikelder2023unified}. The term \textit{consistent} conveys that the NSCH model is established in a consistent manner through mixture theory. Hence, the balance laws naturally emerge from mixture theory, and the model matches the incompressible Navier-Stokes equations in the single fluid regime.

\begin{table}[h!]
{\small
\begin{tabular}{m{15em}m{11em}m{8em}m{10em}}
\textbf{Model}                      & \textbf{MT-consistent BL} & \textbf{Single fluid} & \textbf{Energy-dissipation} \\[6pt] \thickhline\\[-4pt]
Abels et al. \cite{abels2012thermodynamically}  & {\color{red}\xmark}  & {\color{darkgreen}\cmark}      & {\color{orange}\smark}              \\[6pt] 
Aki et al. \cite{aki2014quasi}                  & {\color{darkgreen}\cmark}    & {\color{red}\xmark}    & {\color{darkgreen}\cmark}              \\[6pt] 
Boyer \cite{boyer2002theoretical}                    & {\color{red}\xmark}      & {\color{darkgreen}\cmark}   & {\color{red}\xmark}              \\[6pt] 
Ding et al. \cite{ding2007diffuse}                & {\color{red}\xmark}    & {\color{darkgreen}\cmark}   & {\color{red}\xmark}                   \\[6pt] 
Lowengrub and Truskinovsky \cite{lowengrub1998quasi} & {\color{darkgreen}\cmark}     & {\color{red}\xmark} & {\color{darkgreen}\cmark}                  \\[6pt] 
Shen et al. \cite{shen2013mass}               & {\color{darkgreen}\cmark}     & {\color{red}\xmark} & {\color{darkgreen}\cmark}                \\[6pt] 
Shokrpour Roudbari et al. \cite{shokrpour2018diffuse}  & {\color{darkgreen}\cmark}   & {\color{red}\xmark} & {\color{darkgreen}\cmark}              \\[6pt] 
\textit{Consistent NSCH model} \cite{eikelder2023unified} & {\color{darkgreen}\cmark}   & {\color{darkgreen}\cmark} & {\color{darkgreen}\cmark}              \\[6pt] \hline
\end{tabular}}
\caption{Comparison of the various NSCH models. The column `MT-consistent BL' indicates whether the balance laws (BL) of the model are compatible with mixture theory (MT). In the third column `Single fluid' we state whether the model is compatible with the incompressible Navier-Stokes equation in the single fluid regime, and in the last column whether the model is energy dissipative. The symbol {\color{orange}\smark} indicates that there exists an energy-dissipation law in which the associated kinetic energy is not an obvious approximation of the kinetic energy of the mixture. We refer to ten Eikelder et al. \cite{eikelder2023unified} for details.}
\label{table: overview models constit mod}
\end{table}

We refer from now on to \textit{the NSCH model}, of which a particular form reads:
\begin{subequations}\label{eq: model mass averaged: intro}
  \begin{align}
   \partial_t (\rho \bv) + \divg \left( \rho \bv\otimes \bv \right) + \nabla p + 
   \phi \nabla \mu
   - \divg \boldsymbol{\tau} -\rho\mathbf{g} &=~ 0, \label{eq: model  modified: intro: mom}\\
 \partial_t \rho + \divg(\rho \bv) &=~ 0, \label{eq: model modified: intro: cont} \\
  \partial_t \phi + \divg(\phi \bv)  - \divg \left(\mathbf{M}^v\nabla \left(\mu+\alpha p\right)\right) &=~0,\label{eq: model modified: intro: phi}\\
  \mu  - \dfrac{\partial \Psi}{\partial \phi} + {\rm div}\left( \dfrac{\partial \Psi}{\partial \nabla \phi}\right) &=~0,\label{eq: model modified: intro: mu}
  \end{align}
\end{subequations}
in domain $\Omega \subset \mathbb{R}^d$. Here $\bv$ is the mass averaged velocity, $\rho$ the mixture density, $\nu$ the dynamic viscosity, $p$ the mechanical pressure and $\phi$ the phase variable. 
Furthermore $\mu$ is a chemical potential quantity, $W=W(\phi)$ a double-well potential, $\Psi$ the volumetric free energy, $\mathbf{M}^v$ a degenerate mobility tensor, $\bg$ a force vector, $\alpha$ a constant linked to the density jump, and $\boldsymbol{\tau}$ is the Cauchy stress. The constant surface tension coefficient is $\sigma$ and $\epsilon$ is a parameter associated with the interface thickness. We remark that the reference of this model as a Cahn-Hilliard type model originates from the fact that \eqref{eq: model modified: intro: phi} is a (convective) Cahn-Hilliard equation when selecting the Ginzburg-Landau free energy for $\Psi$.

We emphasize that the NSCH model is a reduced model from the perspective of continuum mixture theory. Namely, the NSCH model consists of a single momentum equation, whereas a full mixture model would contain one momentum equation per constituent. We have recently established a diffuse-interface modeling framework \cite{eikelder2023thermodynamically} that is fully compatible with mixture theory. This model does not contain any Cahn-Hilliard type equation (and thus no associated mobility parameter), however the Allen-Cahn mass transfer model remains present. This accentuates that the Cahn-Hilliard component in the NSCH model emerges from a simplication assumption of NSCH model which does not match with mixture theory.

Over the last decades a large number of numerical methods has been presented for NSCH models with non-matching densities, see e.g. \cite{boyer2002theoretical, ding2007diffuse,shen2013mass,aland2014time,shokrpour2018diffuse,yue2020thermodynamically,bhopalam2022elasto}.
It is well-known in the community that it is hard (or perhaps impossible) to monolitlically discretize a mass-averaged velocity form of the NSCH model. In fact, the authors are not aware of any such numerical methodology. As a consequence, the volume-averaged velocity form of the NSCH model is more popular in the design of numerical methods. This is due to the fact that (in absense of mass transfer) this velocity is divergence-free and one can adopt standard stable velocity/pressure finite element pairs.

In this current article we present the consistent NSCH model along with its first numerical discretization, a verification study, and a validation with experimental data. To this purpose, we transform \eqref{eq: model mass averaged: intro} into an equivalent formulation in terms of a divergence-free velocity. This circumvents the numerical difficulty of the mass-averaged velocity form of the model \eqref{eq: model mass averaged: intro}. We then propose a monolithic discretization methodology that makes use of divergence conforming isogeometric analysis spaces. We note that the usage of these spaces for the discretization of a NSCH model is uncommon but not new, see Espath et al. \cite{espath2016energy}. Finally, we perform a two-dimensional verification, and a three-dimensional validation study.

The remainder of the paper is outlined as follows. In \cref{sec: model} we present the consistent Navier-Stokes Cahn-Hilliard model and analyze its properties. Additionally, we present an alternative but equivalent formulation that forms the basis for the discretization scheme. In \cref{sec: numerical methodology} we introduce the fully-discrete numerical scheme and its properties. Next, in \cref{sec: verification} we compare the proposed method to computations of existing models from literature. Then, in \cref{sec: validation} we simulate a number of three-dimensional benchmark problems. We close the paper with a conclusion and outlook in \cref{sec: summary outlook}.

\section{The Navier-Stokes Cahn-Hilliard model}\label{sec: model}
In this section we present the consistent Navier-Stokes Cahn-Hilliard model. First, in \cref{subsec: gov eq} we present the governing equations, and select compatible constitutive models. Then, in \cref{subsec: physical properties} we discuss the physical properties of the model. In \cref{subsec: non-dim form} we perform the non-dimensionalization.

\subsection{Governing equations and divergence-free formulation}\label{subsec: gov eq}
The derivation of the consistent Navier-Stokes Cahn-Hilliard model relies on mixture theory and the Coleman-Noll procedure. A detailed derivation and discussion on the various modeling choices can be found in ten Eikelder et al. \cite{eikelder2023unified}. 
In this article we work with the NSCH initial/boundary value problem \eqref{eq: model mass averaged: intro} which we repeat in detail: find the mass-averaged velocity $\bv:\Omega \rightarrow \mathbb{R}^d$, the pressure $p:\Omega \rightarrow \mathbb{R}$, the phase field $\phi:\Omega \rightarrow [-1,1]$ and the chemical potential $\mu:\Omega \rightarrow \mathbb{R}$ such that:
\begin{subequations}\label{eq: model mass averaged: sec 2}
  \begin{align}
   \partial_t (\rho \bv) + \divg \left( \rho \bv\otimes \bv \right) + \nabla p + 
   \phi \nabla \mu 
   - \divg \boldsymbol{\tau} -\rho\mathbf{g} &=~ 0, \label{eq: model mass averaged: sec 2: mom}\\
 \partial_t \rho + \divg(\rho \bv) &=~ 0, \label{eq: model mass averaged: sec 2: cont} \\
  \partial_t \phi + \divg(\phi \bv)  - \divg \left(\mathbf{M}^v\nabla \left(\mu+\alpha p\right)\right) &=~0,\label{eq: model mass averaged: sec 2: phi}\\
  \mu  - \dfrac{\partial \Psi}{\partial \phi} + {\rm div}\left( \dfrac{\partial \Psi}{\partial \nabla \phi}\right) &=~0.\label{eq: model mass averaged: sec 2: mu}
  \end{align}
\end{subequations}
with $\bv(\bx,0) = \bv_0(\bx) $ and $\phi(\bx,0) = \phi_0(\bx)$ in $\Omega$. The phase variable $\phi$ represents the difference of the volume fractions of the two constituents; $\phi = 1$ in the first constituent whereas $\phi = -1$ in the second constituent. The density $\rho$ and the dynamic viscosity $\nu$ are the superposition of the constant constituent quantities ($\rho_1, \rho_2$ and $\nu_1, \nu_2$) weighted by their volume fractions:
\begin{subequations}\label{eq: new model rho nu}
\begin{align}
 \rho(\phi) &=   \rho_1\dfrac{1+\phi}{2}+ \rho_2\dfrac{1-\phi}{2},\nn\\
 \nu(\phi)  &= \nu_1 \dfrac{1+\phi}{2} + \nu_2 \dfrac{1-\phi}{2},
\end{align}
\end{subequations}
with constant constituent densities $\rho_1$ and $\rho_2$ and constant constituent viscosities $\nu_1$ and $\nu_2$. Furthermore, $\bg = - g \boldsymbol{\jmath}$, where $\boldsymbol{\jmath}$ is the vertical unit vector and $g$ the gravitational acceleration, and $\alpha = (\rho_2-\rho_1)/(\rho_1+\rho_2)$ is a constant linked to the relative density jump. The Cauchy stress is of the form $\boldsymbol{\tau}=\nu (2\mathbf{D}+\lambda({\rm div}\bv) \mathbf{I})$ where $\mathbf{D}=(\nabla \mathbf{v} + (\nabla \mathbf{v})^T)/2$ is the symmetric velocity gradient, and the factor $\lambda \nu$ represents the second viscosity coefficient. In this work we assume that the Stokes's hypothesis is fulfilled, i.e. $\lambda = -2/d$, where $d$ is the number of spatial dimensions. Next, $\Psi$ denotes the Helmholtz free energy functional belonging to the constitutive class:
\begin{align}\label{eq: class Psi}
  \Psi = \Psi(\phi,\nabla \phi).
\end{align}
The quantity $\mu$ represents a chemical potential-like variable and is defined as the variational derivative of the integral of the Helmholtz free energy $\Psi$, i.e.:
\begin{align}
    \mu := \dfrac{\partial \Psi}{\partial \phi} - \divg \dfrac{\partial \Psi}{\partial \nabla \phi}.
\end{align}
The so-called mobility tensor $\mathbf{M}^v=\mathbf{M}^v(\phi, \nabla \phi, \mu, \nabla \mu, p)$ is a scaling factor of the term $\nabla (\mu+ \alpha p)$. This product is a model for velocity difference of the two components, see \cite{eikelder2023unified}. The mobility tensor is of degenerate type. This means that in the single-fluid regime, i.e. $\phi = \pm 1$, the mobility vanishes: $\mathbf{M}^v=0$. Furthermore, $\mathbf{M}^v$ is compatible with the condition:
\begin{align}\label{eq: mobility condition}
 - \nabla (\mu+ \alpha p) \cdot  \mathbf{M}^v\nabla (\mu + \alpha p) \leq~0.
\end{align}
Equation \eqref{eq: model mass averaged: sec 2: mom} represents the balance of mixture momentum, and \eqref{eq: model mass averaged: sec 2: cont} the balance of mixture mass. Next, \eqref{eq: model mass averaged: sec 2: phi} is the phase field equation, that due to the degenerate type of the mobility tensor, is compatible in the single fluid regime. Lastly, equation defines the variational derivative of the free energy and may be substituted into \eqref{eq: model mass averaged: sec 2: mom} and \eqref{eq: model mass averaged: sec 2: phi}.

We now convert the form \eqref{eq: model mass averaged: sec 2} of the NSCH model by means of the variable transformation:
\begin{align}
   \rho \bv =&~ \rho \bu + \bJ,
\end{align}
where the diffusive flux $\bJ = \bJ(p,\phi,\mu)$, and the degenerate mobility tensor $\mathbf{M}$ are given by:
\begin{subequations}
    \begin{align}
        \bJ =&~ - \dfrac{\rho_1-\rho_2}{2}\mathbf{M}\nabla (\mu+\alpha p),\\
        \mathbf{M}=&~\left(2\rho/\left(\rho_1+\rho_2\right)\right) \mathbf{M}^v.
    \end{align}
\end{subequations}
The (equivalent) NSCH initial/boundary value problem now takes the form: find the volume-averaged velocity $\bu:\Omega \rightarrow \mathbb{R}^d$, the pressure $p:\Omega \rightarrow \mathbb{R}$, the phase field $\phi:\Omega \rightarrow [-1,1]$ and the chemical potential $\mu:\Omega \rightarrow \mathbb{R}$ such that:
\begin{subequations}\label{eq: model orig}
  \begin{align}
   \partial_t (\rho \bu + \bJ) + \divg \left( \rho^{-1} \left(\rho \bu + \bJ\right) \otimes \left(\rho \bu + \bJ\right) \right) + \nabla p + 
   \phi \nabla \mu 
   - \divg  \boldsymbol{\tau} -\rho\mathbf{g} &=~ 0, \label{eq: model orig: mom}\\
 \divg \bu  &=~ 0, \label{eq: model orig: cont} \\
  \partial_t \phi + \bu \cdot \nabla \phi - \divg \left(\mathbf{M}\nabla \left(\mu+\alpha p\right)\right) &=~0,\label{eq: model orig: phi}\\
  \mu  - \dfrac{\partial \Psi}{\partial \phi} + {\rm div}\left( \dfrac{\partial \Psi}{\partial \nabla \phi}\right) &=~0,\label{eq: model orig: mu}
  \end{align}
\end{subequations}
with $\bu(\bx,0) = \bu_0(\bx) $, $\phi(\bx,0) = \phi_0(\bx)$ and $p(\bx,0) = p_0(\bx) $ in $\Omega$. The Cauchy stress $\boldsymbol{\tau} = \boldsymbol{\tau}(\mathbf{u},p,\phi,\mu)$ and symmetric velocity gradient $\mathbf{D}$ take the form:
\begin{subequations}
\begin{align}
   \boldsymbol{\tau} =&~ \nu \left( 2\mathbf{D}+\lambda{\rm div}\left(\rho^{-1}\bJ\right) \mathbf{I}\right),\\
   \mathbf{D} =&~ \nabla \left(\bu + \rho^{-1} \bJ\right)/2 + \nabla \left(\bu + \rho^{-1} \bJ\right)^T/2,
\end{align}
\end{subequations}
The balance of mixture momentum \eqref{eq: model orig: mom} is non-standard due to the form of the inertia terms. The mixture momentum equation of existing volume-averaged velocity models is incomplete when compared with \eqref{eq: model orig: mom}. Namely, these existing models either do not accommodate any diffusive flux ($\bJ$) \cite{ding2007diffuse,boyer2002theoretical}, or contain just a single diffusive flux \cite{abels2012thermodynamically}. The occurrence of the diffusive flux $\bJ$ in the mixture momentum equation is not new. Next, \eqref{eq:  model orig: cont} represents the divergence free property of the velocity field. This key structure of the formulation coincides with the single-fluid incompressible Navier-Stokes equations. Finally, \eqref{eq:  model orig: phi} is the phase field equation, and \eqref{eq:  model orig: mu} defines the variational derivative of the free energy.


In this paper we work with the Helmholtz free energy in the Ginzburg-Landau form:
\begin{subequations}\label{eq: Helmholtz free energy}
\begin{align}
    \Psi =&~ \frac{\sigma}{\varepsilon}W(\phi) + \dfrac{\sigma\varepsilon}{2}|\nabla \phi|^2\\
    W(\phi)=&~\frac{1}{4}(1-\phi^2)^2,
\end{align}
\end{subequations}
where $W=W(\phi)$ represents a double-well potential, $\varepsilon$ represents an interface thickness variable and $\sigma$ is a surface energy density coefficient. Additionally we choose a degenerate isotropic mobility tensor of the form:
\begin{subequations}\label{eq: mob tensor}
  \begin{align}
    \mathbf{M} =&~ m \mathbf{I},\\
    m(\phi) =&~ \gamma  (1- \phi^2)^2,
\end{align}
\end{subequations}
with $\gamma=\gamma(\varepsilon)$. This closes the Navier-Stokes Cahn-Hilliard system which reads:
\begin{subequations}\label{eq: model GL}
  \begin{align}
   \partial_t (\rho \bu + \bJ) + \divg \left(\rho^{-1} \left(\rho \bu + \bJ\right) \otimes \left(\rho \bu + \bJ\right) \right) + \nabla p + 
   \phi \nabla \mu 
   - \divg \boldsymbol{\tau}-\rho\mathbf{g}&=~ 0, \label{eq: model GL: mom}\\
 \divg \bu  &=~ 0, \label{eq: model GL: cont} \\
  \partial_t \phi + \bu \cdot \nabla \phi - \divg \left( m \nabla \left(\mu+\alpha p\right)\right) &=~0,\label{eq: model GL: PF}\\
  \mu - \dfrac{\sigma}{\varepsilon}W'(\phi)+  \sigma \varepsilon \Delta \phi &=~0.
  \end{align}\label{eq: model GL: mu}
\end{subequations}

\begin{rmk}[Korteweg tensor]
In many NSCH models the contribution of the surface forces appears in the momentum equation via a Korteweg type tensor. This is then often subsequently simplified for the Ginzburg-Landau free energy via the identity:
\begin{align}\label{eq: identity free energy}
     \phi\nabla \mu = \divg \left( \sigma \varepsilon \nabla \phi \otimes  \nabla \phi + \left(\mu\phi-\frac{\sigma}{\varepsilon}W(\phi) - \dfrac{\sigma\varepsilon}{2}|\nabla \phi|^2\right)\mathbf{I} \right).
\end{align}
We remark here that this identity has nothing to do with the specific Ginzburg-Landau free energy as it holds in the general case $\Psi = \Psi(\phi,\nabla \phi)$:
\begin{align}
\phi \nabla \mu = \divg \left( \nabla \phi \otimes \dfrac{\partial \Psi}{\partial \nabla \phi} + (\mu\phi-\Psi)\mathbf{I} \right).
\end{align}
\end{rmk}


\subsection{Physical properties}\label{subsec: physical properties}

A physically vital feature of the NSCH model \eqref{eq: model GL} is that it recovers the standard incompressible Navier-Stokes model in the single fluid regime.
\begin{thm}[Reduction to the incompressible Navier-Stokes equations] \label{thm: reduction NS}
  The NSCH model \eqref{eq: model GL} reduces to the standard incompressible Navier-Stokes model in the single fluid regime ($\phi = \pm 1$). 
\end{thm}
\begin{proof}
  Noting that the mobility $m$ is degenerate, the diffusive flux $\mathbf{J}$ vanishes in the single fluid regime. As a consequence, the symmetric velocity gradient reduces to $\mathbf{D} = \left(\nabla \mathbf{u} + (\nabla \mathbf{u})^T\right)/2$. Furthermore, the chemical potential $\mu$ vanishes in the single fluid regime. Taking these observations into account, we find that taking $\phi = \pm 1$ in \eqref{eq: model GL} yields:
  \begin{subequations}\label{eq: model NS}
  \begin{align}
   \partial_t \bu + \divg \left( \bu\otimes \bu \right) + \nabla \tilde{p} - \divg \left( 2\tilde{\nu} \mathbf{D} \right)-\mathbf{g} &=~ 0, \label{eq: model NS: mom}\\
 \divg \bu &=~ 0, \label{eq: model NS: cont}
  \end{align}
\end{subequations}
where $\tilde{p} = p/\rho$ and $\tilde{\nu}=\nu/\rho$ with constant density $\rho$, and $\mathbf{D} = \left(\nabla \mathbf{u} + (\nabla \mathbf{u})^T\right)/2$. 
\end{proof}
\begin{rmk}[Inconsistency single fluid regime]
  Not all exisiting NSCH models reduce to the incompressible Navier-Stokes equations in the single fluid regime. In particular, some models that employ a constant mobility parameter \cite{aki2014quasi,lowengrub1998quasi,shen2013mass,shokrpour2018diffuse} do not share this feature. As a result, inconsistencies can occur in this regime. For details we refer to ten Eikelder et al. \cite{eikelder2023unified}.
\end{rmk}

\begin{thm}[Conservation mixture mass and phase]\label{thm: Conservation}
The formulation conserves the mixture mass and the phase variable:
\begin{subequations}
\begin{align}
        \dfrac{{\rm d}}{{\rm d}t}\displaystyle\int_\Omega \rho ~{\rm d}\Omega =&~ 0, \label{eq: mass conservation global}\\
        \dfrac{{\rm d}}{{\rm d}t}\displaystyle\int_\Omega \phi ~{\rm d}\Omega =&~ 0.
\end{align}
\end{subequations}
\end{thm}
\begin{proof}
This follows from integration of the mixture mass and phase field evolution equations.
\end{proof}
Let us denote the local kinetic, gravitational and global energy as:
\begin{subequations}\label{eq: energy}
  \begin{align}
      \mathscr{K} =&~ \frac{1}{2}\rho \|\bv\|^2 = \frac{1}{2}\rho \| \bu + \rho^{-1} \bJ\|^2, \label{eq: kin energy}\\
      \mathscr{G} =&~ \rho g y,\\
      \mathscr{E} =&~ \displaystyle\int_\Omega \mathscr{K} + \mathscr{G} + \Psi ~{\rm d}\Omega,
  \end{align}
\end{subequations}
where we recall the identity $\bv = \bu + \rho^{-1} \bJ$.
The mass equations and the mixture momentum equation of the NSCH model \eqref{eq: model GL} imply the following global energy evolution.
\begin{thm}[Energy dissipation]\label{thm: energy balance}
Let $\bu, p$ and $\phi$ be smooth solutions of the strong form \eqref{eq: model GL}. The associated total energy $\mathscr{E}$ satisfies the dissipation inequality:
 \begin{align}\label{eq: energy dissipation}
   \frac{{\rm d}}{{\rm d}t} \mathscr{E}=&~- \displaystyle\int_\Omega \nu(\phi)(2\mathbf{D}+ \lambda (\divg \bv) \mathbf{I}):\nabla \bv ~{\rm d}x \nn\\
   &~ - \displaystyle\int_\Omega \nabla (\mu+ \alpha p) \cdot  (m\nabla (\mu + \alpha p)) ~{\rm d}x + \mathscr{B} \leq 0 + \mathscr{B},
 \end{align}
  where $\mathscr{B}$ contains the boundary contributions:
 \begin{align}
     \mathscr{B}=\displaystyle\int_{\partial \Omega} \Bn^T \left(-p\bv + \nu(\phi) (2\mathbf{D}+\lambda({\rm div}\bv) \mathbf{I})\bv-\dfrac{\partial \Psi}{\partial \nabla \phi}\left(\partial_t \phi+ \mathbf{v}\cdot \nabla \phi \right)+   (\mu + \alpha p) m\nabla (\mu + \alpha p)\right)~ {\rm d a},
 \end{align}
 and where we recall $\bv = \bu + \rho^{-1} \bJ$.
 \end{thm}
 \begin{proof}
  See ten Eikelder et al. \cite{eikelder2023unified}.
\end{proof}

\begin{rmk}[Viscous term]
  The viscous term in \eqref{eq: energy dissipation} is negative by means of the well-known identity:
   \begin{align}
    -\nu ( 2 \mathbf{D}+ \lambda \divg \bv \mathbf{I}):\nabla \bv = -2 \nu \left( \mathbf{D} - \frac{1}{d} \left( {\rm div} \bv \right) \mathbf{I} \right) : \left( \mathbf{D} - \frac{1}{d} \left( {\rm div} \bv \right) \mathbf{I} \right) - \nu \left(\lambda + \frac{2}{d} \right) \left( {\rm div} \bv \right)^2 \leq 0.
  \end{align}
\end{rmk}

\begin{rmk}[Second law]
  Even though this energy dissipation property is arguably favorable from the analytical and numerically perspective, it is not equivalent to the second law of thermodynamics. Unfortunately, in many articles the energy dissipation property is incorrectly referred to as the second law. The second law involves individual constituent quantities, whereas \cref{thm: energy balance} contains mixture quantities that do not match the superposition of the constituent quantities. For more details we refer to \cite{eikelder2023unified,eikelder2023thermodynamically}.
\end{rmk}
  \begin{rmk}[Second law]
  The NSCH model is not the sole model for non-matching density flows with an energy-dissipation property. It can be shown that the (diffuse-interface) level-set model (together with a particular discretization) shares this feature \cite{ten2021novel}.
\end{rmk}

Next, we focus on the equilibrium properties. The equilibrium ($E$) solution $(\bu_E,p_E,\phi_E,\mu_E)$ of the model \eqref{eq: model GL} is characterized by:
\begin{align}
    (\bu_E,p_E,\phi_E) =&~ \underset{(\bu,p,\phi)}{\rm argmin}~ \mathscr{E}(\bu,p,\phi), 
\end{align}
where $\mu_E = (\sigma/\varepsilon)W'(\phi_E)-  \sigma \varepsilon \Delta \phi_E$. Restricting to smooth solutions we have the equivalence:
\begin{align}\label{eq: min energy}
    (\bu_E,p_E,\phi_E) = \underset{(\bu,p,\phi,\mu)}{\rm argmin}~ \mathscr{E}(\bu,p,\phi) ~~ \Longleftrightarrow ~~ \dfrac{{\rm d}}{{\rm d}t}\mathscr{E}(\bu_E,p_E,\phi_E) = 0.
\end{align}
Invoking \cref{thm: energy balance} we arrive at the equilibrium conditions:
\begin{subequations}\label{eq: equilibrium cond}
\begin{align}
2 \nu_E \left( \mathbf{D}_E - \frac{1}{d} \left( {\rm div} \bv_E \right) \mathbf{I} \right) : \left( \mathbf{D}_E - \frac{1}{d} \left( {\rm div} \bv_E \right) \mathbf{I} \right) =&~ 0, \label{eq: equilibrium cond a}\\
 \nu_E \left(\lambda + \frac{2}{d} \right) \left( {\rm div} \bv_E \right)^2 =&~0, \label{eq: equilibrium cond b}\\
 m_E \nabla (\mu_E+ \alpha p_E) \cdot  \nabla (\mu_E + \alpha p_E) =&~0, \label{eq: equilibrium cond c}
\end{align}
\end{subequations}
with $\nu_E=\nu(\phi_E)$, $\mathbf{D}_E=\mathbf{D}(\bu_E,\phi_E,\mu_E)$ and $m_E = m(\phi_E)$. From \eqref{eq: equilibrium cond a}-\eqref{eq: equilibrium cond b} we deduce $\bv_E = {\rm const}$, and hence the kinetic energy vanishes. Recalling the condition on the mobility tensor \eqref{eq: mobility condition}, the equilibrium condition \eqref{eq: equilibrium cond c} implies $\mu_E+ \alpha p_E = {\rm const}$. As a consequence, the diffusive flux $\bJ$ vanishes in equilibrium. We now arrive at the equilibrium conditions:
\begin{subequations}
  \begin{align}
     \nabla p_E + \phi_E \nabla \mu_E - \rho_E \mathbf{g} &=~ 0, \label{eq: eq1}\\
     \mu_E+ \alpha p_E &=~ {\rm const}, \label{eq: eq2}
\end{align}
\end{subequations}
where \eqref{eq: eq1} follows from the momentum balance law. In the trivial case of a pure fluid ($\phi = \pm 1$) we have $\mu_E = 0$ and retrieve the hydrostatic equilibrium pressure $p_E = - \rho_E g y + {\rm const}$. We now deduce the equilibrium profile in the non-trivial mixture case ($-1<\phi_E<1$) in absence of gravitational forces ($\mathbf{g} = 0$). Inserting \eqref{eq: eq2} into \eqref{eq: eq1} we find
\begin{align}
    (1-\alpha \phi_E) \nabla p_E = 0.
\end{align}
Since $1-\alpha \phi_E \neq 0$ we find $p_E = {\rm const}$. This is consistent with the observation that fluids at rest have a constant pressure. Subsequently from \eqref{eq: eq2} we find $\mu_E = {\rm const}$. Noting that $\mu_E = 0$ in the pure phase we deduce $\mu_E = 0$. In the one-dimensional situation this condition is fulfilled by the well-known smooth profile:
\begin{align}\label{eq: tanh dim}
    \phi_E(s) = \tanh\left(\frac{s}{\varepsilon\sqrt{2}}\right),
\end{align}
where $s$ denotes the spatial coordinate centered at the interface. One can easily verify that free energy of the profile \eqref{eq: tanh dim} is zero, and thus the equilibrium solution is compatible with \eqref{eq: min energy}. Finally, we associate the surface energy density coefficient $\sigma$ with the surface tension coefficient, denoted $\tilde{\sigma}$. In the one-dimensional situation the integral of the free energy across the interface yields:
\begin{align}
    \int_{\mathbb{R}} \Psi(\phi_E)~ {\rm d}s = \sigma\frac{2\sqrt{2}}{3}.
\end{align}
We use the common practice to associate $\tilde{\sigma}$ with the one-dimensional free energy integral and set $
\tilde{\sigma} = \sigma\frac{2\sqrt{2}}{3}$.

\subsection{Non-dimensional form}\label{subsec: non-dim form}
We re-scale the system \eqref{eq: model orig} based on the following dimensionless variables:
  \begin{align}\label{eq: ref values}
    \bx^* =&~ \frac{\bx}{X_0}, \quad t^* = \frac{t}{T_0}, \quad \bu^* = \frac{\bu}{U_0}, \quad \rho^* = \frac{\rho}{\rho_1},  \quad \nu^* = \frac{\nu}{\nu_1}, \nn\\
    p^* =&~ \frac{p}{\rho_1 U_0^2}, \quad \mu^* = \frac{\mu}{\rho_1 U_0^2}, \quad m^* = \frac{m \rho_1 U_0}{X_0},
  \end{align}
where $X_0, T_0$ and $U_0$ are characteristic length, time and velocity scales, respectively, related via $U_0=X_0/T_0$. The dimensionless system reads:
\begin{subequations}\label{eq: model orig dim less}
  \begin{align}
   \partial_t (\rho \bu + \bJ) + \divg \left( \rho^{-1}\left(\rho \bu +  \bJ\right)\otimes \left(\rho \bu +  \bJ\right) \right) + \nabla p + \phi \nabla \mu
    - \frac{1}{\mathbb{R}{\rm e}}\divg  \boldsymbol{\tau} + \frac{1}{\mathbb{F}{\rm r}^2} \rho\boldsymbol{\jmath} &=~ 0, \label{eq: model orig: mom dim less}\\
   \divg \bu &=~ 0, \label{eq: model orig: cont dim less} \\
 \partial_t \phi +   \bu \cdot \nabla \phi - \divg \left(m \nabla (\mu+\alpha p)\right) &=~0,\label{eq: model orig: PF dim less}\\
  \mu - \frac{1}{\mathbb{W}{\rm e}\mathbb{C}{\rm n}} W'(\phi)+\frac{\mathbb{C}{\rm n}}{\mathbb{W}{\rm e}} \Delta \phi&=~0,\label{eq: model orig: mu dim less}
  \end{align}
\end{subequations}
where we have omitted the $*$ symbols. The dimensionless coefficients are the Reynolds number ($\mathbb{R}{\rm e}$), the Weber number ($\mathbb{W}{\rm e}$), the Froude number ($\mathbb{F}{\rm r}$) and the Cahn number ($\mathbb{C}{\rm n}$) given by:
\begin{subequations}\label{eq: dimensionless quantities}
\begin{align}
     \mathbb{R}{\rm e} =&~ \frac{\rho_1 U_0 X_0}{\nu_1},\\
     \mathbb{W}{\rm e} =&~ \frac{\rho_1 U_0^2 X_0}{\sigma},\\
     \mathbb{F}{\rm r} =&~ \frac{U_0}{\sqrt{g X_0}},\\
     \mathbb{C}{\rm n} =&~ \frac{\varepsilon}{ X_0}. 
\end{align}
\end{subequations}
The kinetic, gravitational, and free energy take the form:
\begin{align}
    \mathscr{K} =&~ \frac{1}{2}\rho \|\bv\|^2 = \frac{1}{2}\rho \| \bu + \rho^{-1} \bJ\|^2, \\
    \mathscr{G} =&~ \frac{1}{\mathbb{F}{\rm r}^2}\rho(\phi) y, \\
    \Psi =&~ \frac{\mathbb{C}{\rm n}}{2\mathbb{W}{\rm e}}  \nabla \phi\cdot\nabla \phi + \frac{1}{\mathbb{W}{\rm e}\mathbb{C}{\rm n}} W(\phi).
\end{align}
The one-dimensional interface profile (in absence of gravitational forces) reads in non-dimensional form:
\begin{align}
    \phi(s) = \tanh\left(\frac{s}{\mathbb{C}{\rm n}\sqrt{2}}\right),
\end{align}
where $s$ is a non-dimensional spatial coordinate centered at the origin.
\section{Numerical methodology}\label{sec: numerical methodology}

In this section we present the numerical method for the consistent NSCH model. First we discuss the weak formulation in \cref{sec: num method; subsec: weak form}. Next, we present the isogeometric spatial discretization in \cref{sec: num method; subsec: spat discr}, and subsequently the temporal discretization in \cref{sec: num method; subsec: time discr}.

\subsection{Weak formulation}\label{sec: num method; subsec: weak form}

We base the methodology on the NSCH model formulated using the volume-averaged velocity as fundamental variable, as presented in the dimensionless form in \eqref{eq: model orig dim less}. We denote the divergence-conforming trial solution space as $\mathcal{W}_T = \mathcal{V}_T \times \mathcal{Q}_T^3$, where $\mathcal{V}_T$ denotes the trail solution space for $\bu=\bu(t)$, and $\mathcal{Q}_T$ for $p=p(t), \phi=\phi(t)$ and $\mu=\mu(t)$. The corresponding divergence-conforming test function space denotes $\mathcal{W} = \mathcal{V} \times \mathcal{Q}^3$. The weak formulation takes the form:\\

\textit{Find $(\bu,p, \phi, \mu) \in \mathcal{W}_T$ such that for all $(\bw, q, \psi, \zeta) \in \mathcal{W}$:}
\begin{subequations}\label{eq: model weak}
\begin{align}
   (\bw,  \partial_t (\rho \bu + \bJ))_\Omega  -(\nabla \bw, \rho^{-1}(\rho \bu + \bJ) \otimes (\rho \bu + \bJ))_\Omega - (\divg \bw, p)_\Omega +\left(\bw, \phi\nabla \mu\right)_\Omega&\nn \\
+ \frac{1}{\mathbb{R}{\rm e}}\left(\nabla \bw, \boldsymbol{\tau} \right)_\Omega + \frac{1}{\mathbb{F}{\rm r}^2} \left(\bw,   \rho \boldsymbol{\jmath}\right)_\Omega &=~ 0, \label{eq: model weak: mom}\\
 ( q, {\rm div} \bu)_\Omega &=~ 0, \label{eq: model weak: cont}\\
  \left( \psi, \partial_t \phi\right)_\Omega +\left(\psi, \bu \cdot \nabla \phi \right)_\Omega + \left( \nabla \psi, m \nabla \left( \mu+\alpha p\right)\right)_\Omega  &=~ 0, \label{eq: model weak: PF}\\
  \left(\zeta, \mu \right)_\Omega  -\left( \nabla \zeta,\frac{\mathbb{C}{\rm n}}{\mathbb{W}{\rm e}}  \nabla \phi\right)_\Omega - \left(\zeta,  \frac{1}{\mathbb{C}{\rm n}\mathbb{W}{\rm e}} W'(\phi)\right)_\Omega  
  &=~0 , \label{eq: model weak: mu}   
\end{align}
\end{subequations}
where $\bJ = - (\rho_1-\rho_2)m\nabla (\mu+\alpha p)/2$. 

\subsection{Spatial discretization}\label{sec: num method; subsec: spat discr}

We apply the finite element/isogeometric analysis methodology. The parametric domain is defined as $\hat{\Omega} := (-1,1)^d \subset \mathbb{R}^d$, and $\mathcal{M}$ denotes the associated mesh. The parametric domain maps via the continuously differentiable geometrical map (with continuously differentiable inverse) $\mathbf{F}:\hat{\Omega} \rightarrow \Omega$ to the physical domain $\Omega \subset \mathbb{R}^d$. Similarly, the parametric mesh $\mathcal{M}$ maps to the corresponding physical mesh via:
\begin{align}
  \mathcal{K} = \mathbf{F}(\mathcal{M}):= \left\{\Omega_K: \Omega_K= \mathbf{F}(Q), Q \in \mathcal{M} \right\}.
\end{align}
Let $\mathbf{J}_{\mathbf{x}} = \partial \bx/\partial \boldsymbol{\xi}$ denote the Jacobian of the mapping $\mathbf{F}$. The parametric mesh size is its diagonal length, $h_Q = \text{diag}(Q)$ for $Q \in \mathcal{M}$, and the associated physical mesh size $h_K$ is:
\begin{align}
  h_{K}^2 := \frac{h_Q^2}{d} \| \mathbf{J}_{\mathbf{x}} \|_F^2.
\end{align}
Here we recall that $d$ denotes the number of  physical dimensions, and we use the subscript $F$ to denote the Frobenius norm. The (objective) Frobenius norm of the Jacobian equals:
\begin{align}
  \| \mathbf{J}_{\mathbf{x}} \|_F^2 = {\rm Tr}\left(\mathbf{G}^{-1}\right),
\end{align}
where ${\rm Tr}$ denotes the trace operator. Here the element metric tensor and its inverse are:
\begin{subequations}
\begin{align}
  \mathbf{G} =&~ 
  \mathbf{J}_{\mathbf{x}}^{-T} \mathbf{J}_{\mathbf{x}}^{-1},\\
  \mathbf{G}^{-1} =&~ 
  \mathbf{J}_{\mathbf{x}} \mathbf{J}_{\mathbf{x}}^{T}.
\end{align}
\end{subequations}

We now introduce the discrete isogeometric test function space $\WW^h \subset \WW$ and time-dependent solution space $\WW^h_T \subset \WW_T$ spanned by NURBS basis functions. The superscript $h$ indicates that the space is finite-dimensional.  Both the test function space and solution space are divergence-conforming, and we take $\mathcal{W}^h:= \mathcal{V}^h \times (\mathcal{Q}^h)^3$ and $\mathcal{W}_T^h:= \mathcal{V}_T^h \times (\mathcal{Q}_T^h)^3$. For details of the construction we refer to Evans and Hughes \cite{Evans13unsteadyNS}. Applying the continuous Galerkin method now results in the semi-discrete approximation of \eqref{eq: model GL}:\\

\textit{Find $(\bu^h, p^h, \phi^h, \mu^h) \in \WW^h_T$ such that for all $(\bw^h, q^h, \psi^h, \zeta^h) \in \WW^h$:}
\begin{subequations}\label{eq: model weak semi-discrete}
\begin{align}
   (\bw^h,  \partial_t (\rho^h \bu^h + \bJ^h))_\Omega  -(\nabla \bw^h, (\rho^h)^{-1} (\rho^h \bu^h + \bJ^h) \otimes (\rho^h \bu^h + \bJ^h))_\Omega   &\nn \\
- (\divg \bw^h, p)_\Omega+\left(\bw^h, \phi^h\nabla \mu^h\right)_\Omega+ (\nabla \bw^h, \boldsymbol{\tau}^h)_\Omega+ \frac{1}{\mathbb{F}{\rm r}^2} \left(\bw^h, \rho^h \boldsymbol{\jmath}\right)_\Omega &=~ 0, \label{eq: model weak semi-discrete: mom}\\
 ( q^h, {\rm div} \bu^h)_\Omega &=~ 0, \label{eq: model weak semi-discrete: cont}\\
  \left( \psi^h, \partial_t \phi^h\right)_\Omega + \left(\psi^h, \bu^h \cdot \nabla \phi^h \right)_\Omega + \frac{1}{\mathbb{W}{\rm e}}\left( \nabla \psi, m\nabla \left(
  \mu^h + \alpha p^h\right)\right)_\Omega  &=~ 0, \label{eq: model weak semi-discrete: PF}\\
  \left(\zeta^h, \mu^h \right)_\Omega  -\left( \nabla \zeta^h,\frac{\mathbb{C}{\rm n}}{\mathbb{W}{\rm e}}  \nabla \phi^h\right)_\Omega - \left(\zeta^h,  \frac{1}{\mathbb{C}{\rm n}\mathbb{W}{\rm e}} W'(\phi^h)\right)_\Omega  
  &=~0, \label{eq: model weak semi-discrete: mu}
\end{align}
\end{subequations}
where $\bu^h(0) = \bu^h_0$, $\phi^h(0)=\phi^h_0$ and $p^h(0)=p^h_0$ in $\Omega$, $\rho^h=\rho(\phi^h)$, $\boldsymbol{\tau}^h=\boldsymbol{\tau}(\mathbf{u}^h,p^h,\phi^h,\mu^h)$, $\bJ^h = \bJ(p^h,\phi^h,\mu^h)$ and $m^h=m(\phi^h)$. The semi-discrete formulation \eqref{eq: model weak semi-discrete} reduces to the standard (weak) form of the incompressible Navier-Stokes equations in the single fluid regime. Additionally, it inherits the conservation property from \cref{thm: Conservation} and has pointwise divergence-free velocities.
\begin{thm}[Properties semi-discrete formulation]
\label{thm: energy dissipation modified model cont}
 Let $(\bu^h, p^h, \phi^h, \mu^h)$ be a smooth solution of the semi-discrete formulation \eqref{eq: model weak semi-discrete}. 
 The formulation has the properties:
 \begin{enumerate}
  \item It conserves the global mass and the global phase field:
  \begin{subequations}
  \begin{align}
    \dfrac{{\rm d}}{{\rm d}t}\displaystyle\int_\Omega \rho^h ~{\rm d}\Omega =&~ 0,\\
    \dfrac{{\rm d}}{{\rm d}t}\displaystyle\int_\Omega \phi^h ~{\rm d}\Omega =&~ 0.
  \end{align}
  \end{subequations}
  \item It has pointwise divergence-free velocities:
  \begin{align}\label{prop 2 cont}
    \divg \bu^h \equiv 0.
  \end{align}
 \end{enumerate}
\end{thm}
\begin{proof}
1. Phase conservation follows from taking $\psi^h = 1$ in the phase field equation \eqref{eq: model weak semi-discrete: PF}. Mass conservation subsequently follows from the observation that $\rho^h$ is an affine function of $\phi^h$.\\

2. Since $\nabla \cdot \bu^h \in \mathcal{Q}^h_T$ we can substitute this weighting function choice into \eqref{eq: model weak semi-discrete: cont}:
\begin{align}
0= \left(  \divg \bu^h, \divg \bu^h\right)_\Omega \quad \Rightarrow \quad  \divg \bu^h \equiv 0 \quad \text{in}~\Omega.
\end{align}
\end{proof}

\begin{rmk}[Energy-dissipation]
  The formulation \eqref{eq: model weak semi-discrete} does not inherit the energy-dissipation property of \cref{thm: energy balance}. This is a consequence of the observation that the weighting function that would lead to an energy-dissipation statement is not a member of $\WW^h$.
\end{rmk}

\subsection{Temporal discretization}\label{sec: num method; subsec: time discr}
To introduce the time-discretization we first subdivide the time domain $\mathcal{T}$ into elements $\mathcal{T}_n=(t_n,t_{n+1})$ of size $\Delta t_n=t_{n+1}-t_n$ with time level $n=0,1,...,N$. We make use of the conventional notation of indicating the time level of a discrete quantity with a subscript, e.g. $\bu^h_n, p^h_n, \phi^h_n$ and $\mu^h_n$ denote the discrete velocity $\mathbf{u}^h$, pressure $p^h$, phase field $\phi^h$ and chemical potential $\mu^h$ at time level $n$. We write $[\![\mathbf{a}^h]\!]_{n} := \mathbf{a}^h_{n+1}-\mathbf{a}^h_{n}$ for the jump of the vector quantity $\mathbf{a}^h$. The intermediate time-levels and associated time derivatives are given by:
\begin{subequations}\label{eq:  intermediate time-levels}
\begin{align}
  \bu^h_{n+1/2} :=&~ \tfrac{1}{2}(\bu^h_n+\bu^h_{n+1}), & \frac{1}{\Delta t_n}[\![\bu^h]\!]_{n} :=&~ \frac{1}{\Delta t_n}(\bu^h_{n+1}-\bu^h_{n}),\\
  \phi^h_{n+1/2} :=&~ \tfrac{1}{2}(\tilde{\phi}^h_n+\phi^h_{n+1}), & \frac{1}{\Delta t_n}[\![\phi^h]\!]_{n} :=&~ \frac{1}{\Delta t_n}(\phi^h_{n+1}-\tilde{\phi}^h_{n})\\
 \rho^h_{n+1/2} :=&~ \rho(\phi^h_{n+1/2}), & \frac{1}{\Delta t_n}[\![\rho^h]\!]_{n} :=&~ \frac{1}{\Delta t_n}(\rho^h_{n+1}-\rho^h_{n}),\\
 \bJ^h_{n+1/2} :=&~ \bJ(p^h_{n+1/2},\phi^h_{n+1/2},\mu^h_{n+1/2}), &\frac{1}{\Delta t_n}[\![\rho^h\bu^h]\!]_{n} :=&~ \frac{1}{\Delta t_n}\left(\rho^h_{n+1}  \bu^h_{n+1}-\rho^h_{n}  \bu^h_{n}\right),\\
 &&\frac{1}{\Delta t_n}[\![\bJ^h]\!]_{n} :=&~ \frac{1}{\Delta t_n}\left(\bJ^h_{n+1}- \bJ^h_{n}\right),
\end{align}
\end{subequations}
where we define $\tilde{\phi}^h_n = \phi^h_n$ for $|\phi^h_n| \leq 1$, and $\tilde{\phi}^h_n = 1$ (respectively $-1$) for $\phi^h_n >1$ (respectively $\phi^h_n < -1$). This permits working with large density and viscosity ratios.  We additionally define
\begin{subequations}
\begin{align}
\rho^h_{n}:=&~ \rho(\tilde{\phi}^h_{n}), &\rho^h_{n+1} =&~  \rho(\phi^h_{n+1})\\
m^h_{n}:=&~ m(\tilde{\phi}^h_{n}), & m^h_{n+1}:=&~  m(\phi^h_{n+1})\\
\bJ^h_{n} :=&~ \bJ(p^h_n,\tilde{\phi}^h_n,\mu^h_n), &\bJ^h_{n+1} :=&~  \bJ(p^h_{n+1},\phi^h_{n+1},\mu^h_{n+1}),\\
{W'}^h_{n+1/2} =&~ {W'}(\phi_{n+1/2}^h)&&\\
 \mathbf{D}^h_{n+1/2} =&~ \nabla^s \left(\bu^h_{n+1/2} + (\rho^h_{n})^{-1} \bJ^h_{n+1/2}\right)\\
\boldsymbol{\tau}^h_{n+1/2}=&~ \nu(\phi^h_{n+1/2}) \left( 2\mathbf{D}^h_{n+1/2}+\lambda{\rm div}\left((\rho^h_{n})^{-1}\bJ^h_{n+1/2}\right) \mathbf{I}\right),
\end{align}
\end{subequations}
Furthermore, the separate pressure and chemical potential terms are taken at time level $n+1$. The temporal discretization of the remaining terms uses the midpoint scheme. The method in fully-discrete form now reads:\\

\textit{Given $\bu_n^h, p^h_n, \phi^h_n$ and $\mu_n^h$, find $\bu_{n+1}^h, p_{n+1}^h, \phi_{n+1}^h$ and $\mu_{n+1}^h$ such that for all $(\bw^h, q^h, \psi^h, \zeta^h) \in \WW_{0,h}$:}
\begin{subequations}\label{eq: fully discrete algorithm}
\begin{align}
\left(\bw^h,  \dfrac{[\![\rho^h\bu^h+\mathbf{J}^h]\!]_{n}}{\Delta t_n} \right)_\Omega - (\divg \bw^h, p^h_{n+1})_\Omega +\left(\bw^h, \phi^h_{n+1/2} \nabla \mu^h_{n+1}\right)_\Omega &\nn\\
-(\nabla \bw^h, \rho^h_{n+1/2} \bu^h_{n+1/2} \otimes  \bu^h_{n+1/2} + \bu^h_{n+1/2} \otimes  \bJ^h_{n+1/2} + \bJ^h_{n+1/2} \otimes  \bu^h_{n+1/2})_\Omega &\nn\\
+ ((\rho^h_{n})^{-1}\bJ^h_{n+1/2} \otimes  \bJ^h_{n+1/2})_\Omega 
+ (\nabla \bw^h, \boldsymbol{\tau}^h_{n+1/2})_\Omega +\frac{1}{\mathbb{F}{\rm r}^2}(\bw^h, \rho^h_{n+1/2} \boldsymbol{\jmath})_\Omega&=~ 0, \label{eq: fully discrete algorithm: mom}\\
  ( q^h, {\rm div} \bu^h_{n+1/2})_\Omega &=~ 0, \label{eq: fully discrete algorithm: continuity}\\ 
 \left( \psi^h, \dfrac{[\![\phi^h]\!]_{n}}{\Delta t_n}  \right)_\Omega + ( \psi^h,  \bu^h_{n+1/2} \cdot\nabla\phi^h_{n+1/2} )_\Omega + \frac{1}{\mathbb{W}{\rm e}}\left( \nabla \psi^h, m^h_{n+1/2}\nabla \left(
  \mu^h_{n+1} + \alpha p^h_{n+1}\right)\right)_\Omega &=~ 0, \label{eq: fully discrete algorithm: phi}\\
  \left(\zeta^h, \mu^h_{n+1}\right)_\Omega- \frac{\mathbb{C}{\rm n}}{\mathbb{W}{\rm e}}\left(\nabla \zeta^h, \nabla \phi^h_{n+1/2}\right)_\Omega - \frac{1}{\mathbb{C}{\rm n}\mathbb{W}{\rm e}} \left(\zeta^h, {W'}^h_{n+1/2}\right)_\Omega &=~ 0 \label{eq: fully discrete algorithm: mu}.
\end{align}
\end{subequations}

\begin{thm}[Fully-discrete divergence-free velocities]\label{thm: fully-discrete}
 The algorithm \eqref{eq: fully discrete algorithm} has pointwise divergence-free solutions:
  \begin{align}\label{eq: fully discrete div conf}
    \divg \bu^h_{n+1/2} \equiv 0.
  \end{align}
\end{thm}
The proof is analogously to the semi-discrete case.

\clearpage

\section{Comparison with existing models}\label{sec: verification}

In this section we verify the computational setup via two-dimensional buoyancy-driven rising bubble problems. The simulation of bubble dynamics problems involves all aspects of the Navier-Stokes Cahn-Hilliard model (i.e. inertia, viscous forces, gravity, and surface tension effects).
In this benchmark problem a bubble of (lighter) fluid $2$ with initial diameter $D_0 = 2R_0 = 0.5$ is placed in the rectangular domain $[0,1]\times[0,2]$ at location $(0.5,0.5)$ in (heavier) fluid $1$ \cite{hysing2009quantitative}. The initial phase field profile is:
\begin{align}\label{eq: init phi 2D}
  \phi^h_0(\mathbf{x}) = \tanh{\dfrac{\sqrt{(x-0.5)^2+(y-0.5)^2}-R_0}{\mathbb{C}{\rm n}\sqrt{2}}}.
\end{align}
At the left and right boundaries a no-penetration boundary condition ($\mathbf{u}\cdot \mathbf{n} = 0$) is applied, and at the top and bottom boundary a no-slip boundary condition ($\mathbf{u} = 0$) is enforced. A sketch of the problem setup is given in \cref{fig:sketch 2D rising bubble problem}.

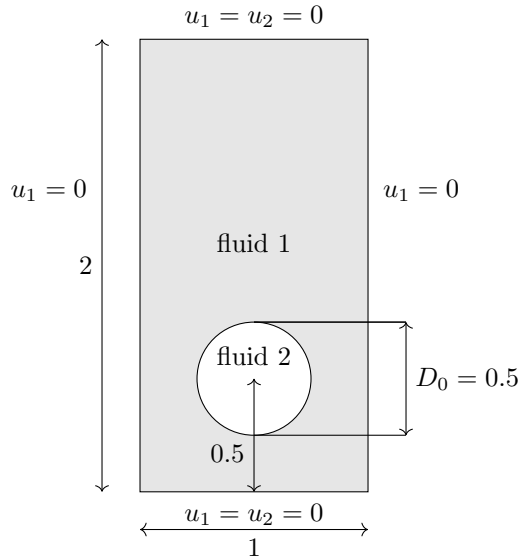
\begin{figure}[h]
\begin{center}
\begin{tikzpicture}
    \fill[fill=gray!20,draw=black] (0, 0) rectangle (3, 6);
    
    \draw[fill=white,draw=black] (1.5, 1.5) circle (0.75);
    
    \draw[<->] (0, -0.5) -- (3, -0.5) node[midway,below] {$1$};
    \draw[<->] (-0.5, 0) -- (-0.5, 6) node[midway,left] {$2$};
    \draw[<->] (3.5, 0.75) -- (3.5, 2.25) node[midway,right] {$D_0 = 0.5$};
    \draw[-] (1.5, 0.75) -- (3.5, 0.75);
    \draw[-] (1.5, 2.25) -- (3.5, 2.25);
    \draw[<->] (1.5, 0.0) -- (1.5, 1.5) node[midway,below left] {$0.5$};
    \draw[-] (1.5, 2.25) -- (3.5, 2.25);
    \node at (1.5, 6.3) {$u_1 = u_2 = 0$}; 
    \node at (1.5, -0.3) {$u_1 = u_2 = 0$};
    \node at (-1.2, 4.0) {$u_1 = 0$}; 
    \node at (3.7, 4.0) {$u_1 = 0$}; 
    \node at (1.5, 1.8) {fluid 2};
    \node at (1.5, 3.3) {fluid 1}; 
\end{tikzpicture}
    \caption{Situation sketch cases 1 and 2}
    \label{fig:sketch 2D rising bubble problem}
\end{center}
\end{figure}

The motion of rising bubble problems is typically described by the Archimedes number ($\mathbb{A}{\rm r}$) and the E\"{o}tv\"{o}s number ($\mathbb{E}{\rm o}$) since the Reynolds number is a priori unknown. The Archimedes number measures the relative importance of buoyancy to viscous forces, and the E\"{o}tv\"{o}s number describes the relative importance of gravity and surface tension. The E\"{o}tv\"{o}s number is also known as the Bond number. 
We select as reference values (see \eqref{eq: ref values}):
\begin{subequations}
    \begin{align}
        X_0 =&~ D_0 \\
        U_0 =&~ \dfrac{D_0}{T_0} \\
        T_0 =&~ \sqrt{\dfrac{\rho_1 D_0^3}{\sigma}},
    \end{align}
\end{subequations}
where $T_0$ is the capilary time scale. As a consequence the dimensionless numbers (\cref{subsec: non-dim form}) become:
\begin{subequations}
    \begin{align}
        \mathbb{R}{\rm e} =&~ \mathbb{E}{\rm o}^{-1/2}\mathbb{A}{\rm r},\\
        \mathbb{F}{\rm r} =&~ \mathbb{E}{\rm o}^{-1/2},\\
        \mathbb{W}{\rm e} =&~ 1,
    \end{align}
\end{subequations}
where the Archimedes number ($\mathbb{A}{\rm r}$) and the E\"{o}tv\"{o}s number ($\mathbb{E}{\rm o}$) are given by:
\begin{subequations}\label{eq: dimensionless quantities 2}
\begin{align}
     \mathbb{A}{\rm r} =&~ \dfrac{\rho_1\sqrt{g D_0^3}}{\nu_1},\\
     \mathbb{E}{\rm o} =&~ \frac{\rho_1 g D_0^2}{\sigma}. 
\end{align}
\end{subequations}
The system is now characterized by $5$ dimensionless quantities: $\mathbb{A}{\rm r}$, $\mathbb{E}{\rm o}$, $\mathbb{C}{\rm n}$, $\rho_1/\rho_2$, $\nu_1/\nu_2$. The benchmark problem involves two cases which are described by different parameter values in \cref{table: parameters 2D RB cases}. All computations were performed on a rectangular uniform mesh with physical element sizes $h = h_K = 1/16, 1/32, 1/64, 1/128$. The computations employ basis functions that are mostly $C^0$-linear, however every velocity space is enriched to be quadratic $C^1$ in the associated direction. The time step size is taken as $\Delta t_n = 0.128 h$, and the Cahn number as $\mathbb{C}{\rm n} = 1.28 h$.

\begin{table}[htbp]
\centering
\begin{tabularx}{\textwidth}{XXXXXXXXX}
Case & \hspace{0.1cm} $\rho_1$ & \hspace{0.1cm} $\rho_2$ & $\mu_1$ & $\mu_2$ & \hspace{0.1cm} $\sigma$ & \hspace{0.1cm} $g$ & $\mathbb{A}{\rm r}$ & $\mathbb{E}{\rm o}$ \\[4pt]
\hline\\[-6pt]
\hspace{0.5cm}1 & $1000$ & $100$ & $10$ & $1$   & $24.5$ & $0.98$ & $35$ & \hspace{0.05cm} $10$   \\[6pt]
\hspace{0.5cm}2 & $1000$ & \hspace{0.1cm} $1$   & \hspace{0.1cm} $1$  & $0.1$ &$1.96$  & $0.98$ & $35$ & $125$  \\[6pt]
\hline
\end{tabularx}
\caption{Parameters for the two-dimensional rising bubble cases.}
\label{table: parameters 2D RB cases}
\end{table}

\cref{fig: case 1 contours,fig: case 2 contours} show the zero phase field ($\phi=0$) contours for cases 1 and 2, respectively. We see that the deformation of the bubble is rather small in case 1, whereas in case 2 we observe significant deformation. In both cases, there are almost no visible differences between the results of the finest two meshes. For a quantitative comparison with reference results from the literature, we use the the center of mass ($y_b$) and rise velocity ($v_b$) defined as:
\begin{subequations}
    \begin{align}
        y_b :=&~ \dfrac{\int_{\phi < 0} y ~{\rm d}x}{\int_{\phi < 0} ~ {\rm d}x},\\
        v_b :=&~ \dfrac{\int_{\phi < 0} u_2 ~{\rm d}x}{\int_{\phi < 0} ~ {\rm d}x}.
    \end{align}
\end{subequations}

\begin{figure}[!ht]
\captionsetup[subfigure]{justification=centering}
\begin{subfigure}{0.49\textwidth}
\centering
\includegraphics[width=0.65\textwidth]{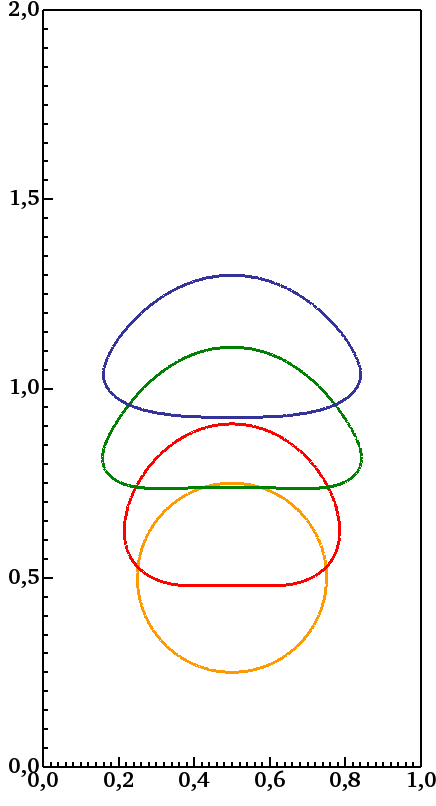}
\caption{$h = 1/128$, $t=0, 1, 2, 3$\\ in orange, red, green and blue (resp).}
\end{subfigure}
\begin{subfigure}{0.49\textwidth}
\centering
\includegraphics[width=0.65\textwidth]{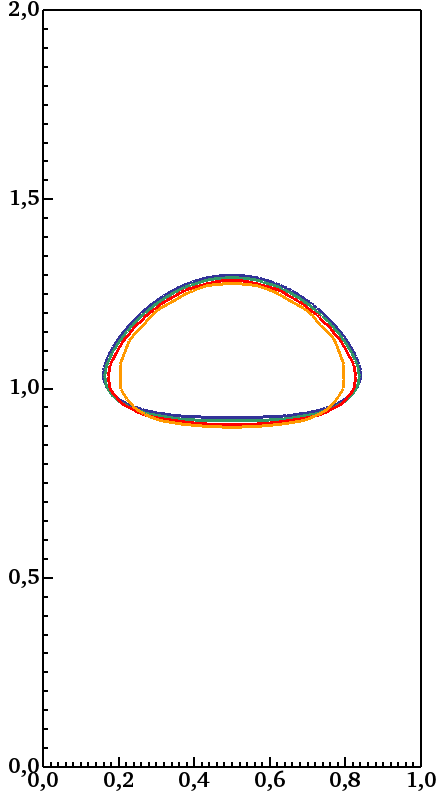}
\caption{$t = 3$, $h = 1/16, 1/32, 1/64, 1/128$\\  in orange, red, green and blue (resp).}
\end{subfigure}
\caption{Case 1. Contours of the phase field $\phi=0$, (a) different time instances, (b) different mesh sizes.}
\label{fig: case 1 contours}
\end{figure}

\newpage
\begin{figure}[!ht]
\captionsetup[subfigure]{justification=centering}
\begin{subfigure}{0.49\textwidth}
\centering
\includegraphics[width=0.65\textwidth]{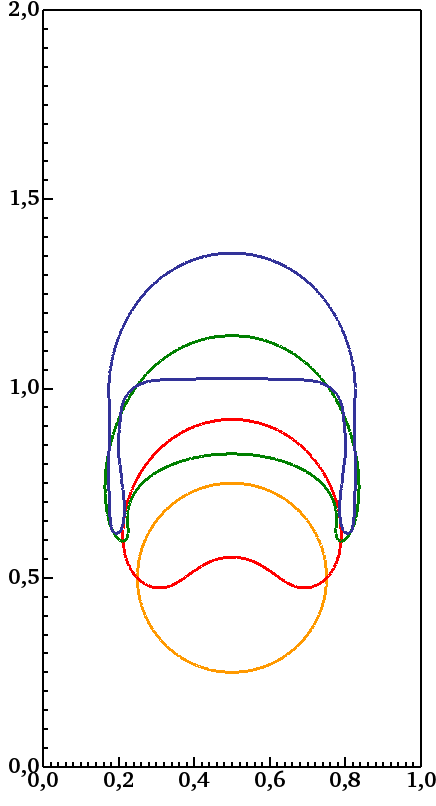}
\caption{$h_K = 1/128$, $t=0, 1, 2, 3$\\ in orange, red, green and blue (resp).}
\end{subfigure}
\begin{subfigure}{0.49\textwidth}
\centering
\includegraphics[width=0.65\textwidth]{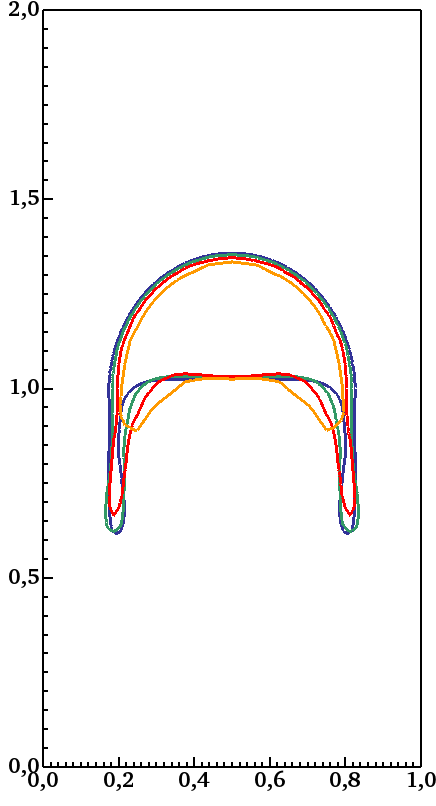}
\caption{$t = 3$, $h_K = 1/16, 1/32, 1/64, 1/128$\\  in orange, red, green and blue (resp).}
\end{subfigure}
\caption{Case 2. Contours of the phase field $\phi=0$, (a) different time instances, (b) different mesh sizes.}
\label{fig: case 2 contours}
\end{figure}

In \cref{fig: case 1 CoM,fig: case 1 RV,fig: case 2 CoM,fig: case 2 RV} we plot for the two cases (i) the center of mass and (ii) the rise velocity, for different mesh sizes, and relative to reference computational data. This data is obtained with (1) the TP2D code (2) the FreeLIFE code, (3) the MooNMD code, see  \cite{hysing2009quantitative}, and the NSCH models of Abels et al. \cite{abels2012thermodynamically}, Boyer \cite{boyer2002theoretical} and Ding et al. \cite{ding2007diffuse}. The computations with the NSCH models were performed by Aland and Voigt \cite{aland2012benchmark}. The center of mass matches well with the reference data for both cases. Concerning the rise velocity, we observe significant difference  for case 2 for $t>1.5$. In this regime, our computational results agree quite well with the NSCH computations performed by Aland and Voigt \cite{aland2012benchmark}, but not with the TP2D, FreeLIFE and MooNMD code results. We remark that both the physical models as well as the computational methods differ between the computational results. The NSCH model is an energy stable model with a diffuse interface, whereas the reference data of the TP2D, FreeLIFE and MooNMD codes is based on a level set description of the interface.  

\begin{figure}[!ht]
\begin{subfigure}{0.49\textwidth}
\centering
\includegraphics[width=0.95\textwidth]{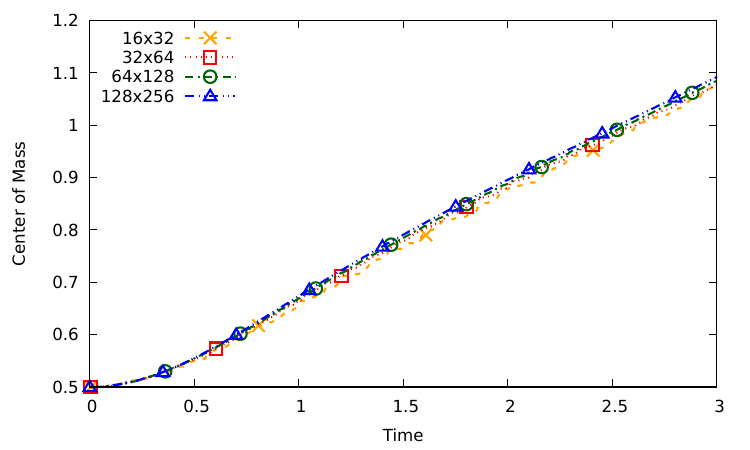}
\caption{$h_K = 1/16, 1/32, 1/64, 1/128$.}
\end{subfigure}
\begin{subfigure}{0.49\textwidth}
\centering
\includegraphics[width=0.95\textwidth]{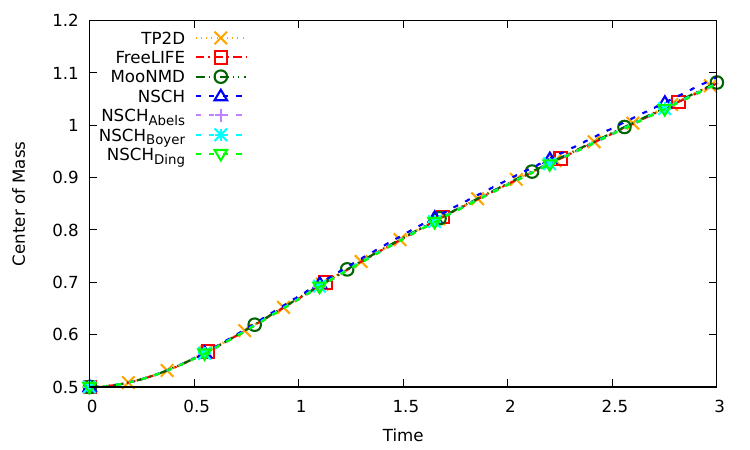}
\caption{Relative to reference data.}
\end{subfigure}
\caption{Case 1. Center of mass (a) for different mesh sizes, and (b) a comparison of the finest mesh results to reference data.}
\label{fig: case 1 CoM}
\end{figure}

\clearpage

\begin{figure}[!ht]
\begin{subfigure}{0.49\textwidth}
\centering
\includegraphics[width=0.95\textwidth]{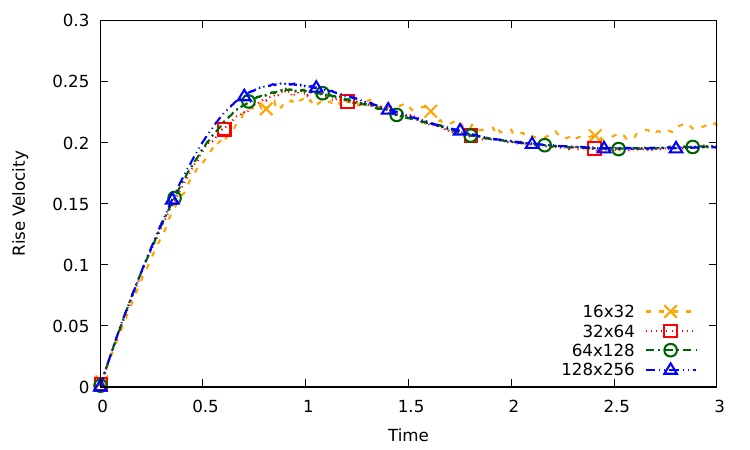}
\caption{$h_K = 1/16, 1/32, 1/64, 1/128$.}
\end{subfigure}
\begin{subfigure}{0.49\textwidth}
\centering
\includegraphics[width=0.95\textwidth]{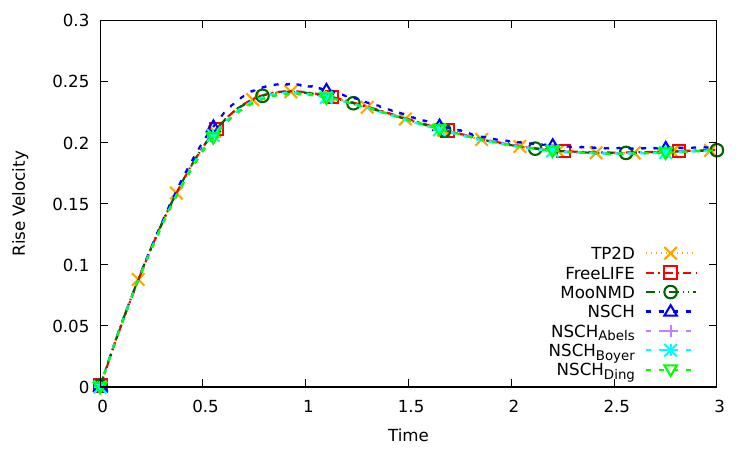}
\caption{Relative to reference data.}
\end{subfigure}
\caption{Case 1. Rise velocity (a) for different mesh sizes, and (b) a comparison of the finest mesh results to reference data.}
\label{fig: case 1 RV}
\end{figure}


\begin{figure}[!ht]
\begin{subfigure}{0.49\textwidth}
\centering
\includegraphics[width=0.95\textwidth]{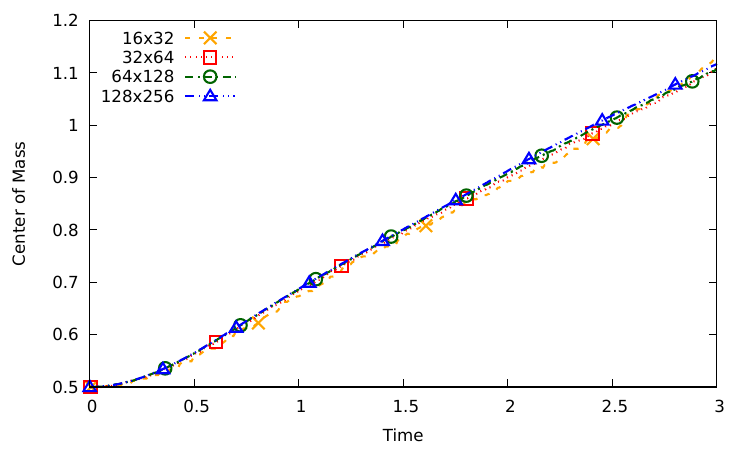}
\caption{$h = 1/16, 1/32, 1/64, 1/128$.}
\end{subfigure}
\begin{subfigure}{0.49\textwidth}
\centering
\includegraphics[width=0.95\textwidth]{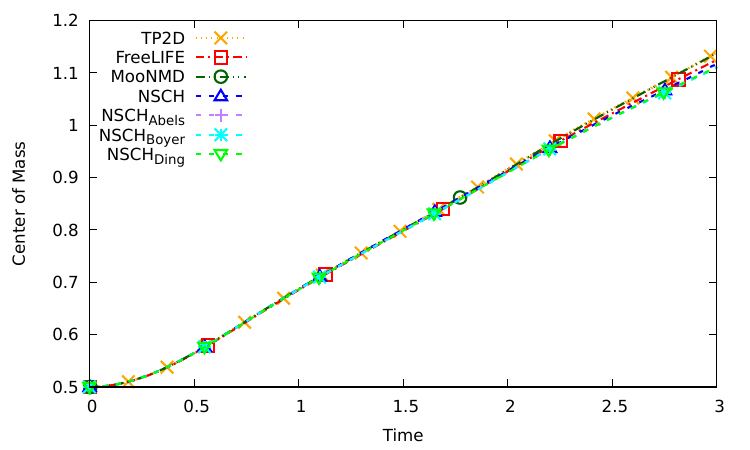}
\caption{Relative to reference data.}
\end{subfigure}
\caption{Case 2. Center of mass (a) for different mesh sizes, and (b) a comparison of the finest mesh results to reference data.}
\label{fig: case 2 CoM}
\end{figure}

\begin{figure}[!ht]
\begin{subfigure}{0.49\textwidth}
\centering
\includegraphics[width=0.95\textwidth]{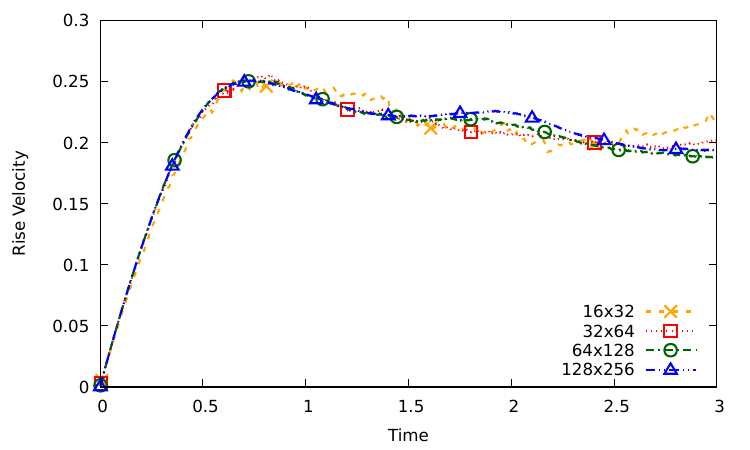}
\caption{$h_K = 1/16, 1/32, 1/64, 1/128$.}
\end{subfigure}
\begin{subfigure}{0.49\textwidth}
\centering
\includegraphics[width=0.95\textwidth]{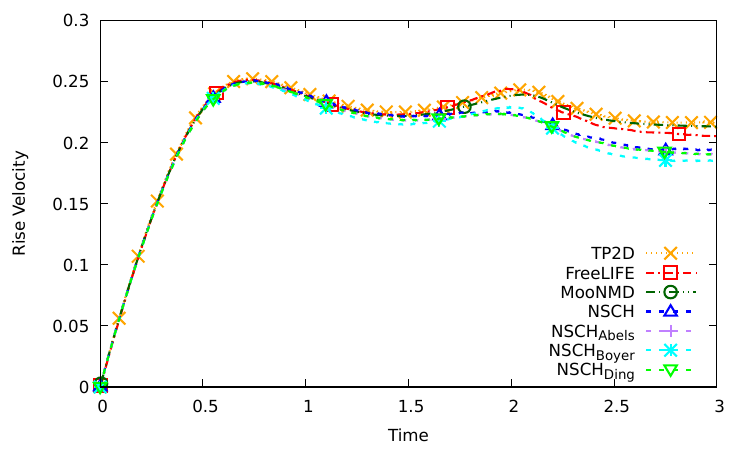}
\caption{Relative to reference data.}
\end{subfigure}
\caption{Case 2. Rise velocity (a) for different mesh sizes, and (b) a comparison of the finest mesh results to reference data.}
\label{fig: case 2 RV}
\end{figure}




\clearpage 

\section{Numerical benchmarks in three dimensions}\label{sec: validation}

In this section we simulate two benchmark problems to validate our consistent NSCH model against experimental data. We first study a buoyancy-driven rising bubble, and second the contraction of a liquid filament. Analogously to \cref{sec: verification}, all the implementations use basis functions that are mostly $C^0$-linear, but every velocity space is enriched to be quadratic $C^1$ in the associated direction. The system of equations is solved with the standard GMRES method with additive Schwartz preconditioning
provided by Petsc \cite{petsc-efficient}.

\subsection{Three-dimensional buoyancy-driven rising bubbles}

The simulation of the three dimensional rising bubble problem is similar to the two-dimensional setup in \cref{sec: verification}.
In particular we use the same reference values and definitions of the dimensionless parameters. In this validation case we consider an air bubble at $20^{\circ}{\rm C}$ in the water. The deforming bubble rises due to buoyancy, and after some time takes its final shape and velocity. The physical domain setup coincides with that of Yan et al. \cite{yan2019isogeometric}: the bubble with initial diameter $D_0 = 2 R_0 = 1$ is placed in the rectangular domain $[0,12]\times[0,24]\times[0,12]$ at location $(6,10.5,6)$. The initial phase field profile is:
\begin{align}\label{eq: init phi 3D RB}
  \phi^h_0(\mathbf{x}) = \tanh{\dfrac{\sqrt{(x-6)^2+(y-10.5)^2 + (z-6)^2}-R_0}{\mathbb{C}{\rm n}\sqrt{2}}}.
\end{align}
At all boundaries a no-penetration boundary condition ($\mathbf{u}\cdot \mathbf{n} = 0$) is applied. A sketch of the problem setup is given in \cref{fig:sketch 3D rising bubble problem: setup}.

\begin{figure}[!h]
    \centering
\begin{subfigure}{0.67\textwidth}
\begin{tikzpicture}[scale=4, transform shape]
    \coordinate (A) at (0,0,0);
    \coordinate (B) at (1,0,0);
    \coordinate (C) at (1,0,1);
    \coordinate (D) at (0,0,1);
    \coordinate (E) at (0,2,0);
    \coordinate (F) at (1,2,0);
    \coordinate (G) at (1,2,1);
    \coordinate (H) at (0,2,1);

    \fill[gray!20] (A) -- (B) -- (C) -- (D) -- cycle;
    \fill[gray!20] (E) -- (F) -- (G) -- (H) -- cycle;
    \fill[gray!20] (B) -- (F) -- (G) -- (C) -- cycle;
    \fill[gray!20] (A) -- (E) -- (H) -- (D) -- cycle;
    \fill[gray!20] (A) -- (B) -- (F) -- (E) -- cycle;
    
    \draw [line width=0.35mm, red!70!black!80, dashed] (A) -- (B) ;
    \draw [line width=0.35mm, red!70!black!80, dashed] (A) -- (D) ;
    \draw [line width=0.35mm, red!70!black!80, dashed] (A) -- (E);
    \draw [line width=0.25mm, black] (D) -- (C) ;
    \draw [line width=0.25mm, black] (C) -- (B) ;
    \draw [line width=0.25mm, black] (E) -- (F) -- (G) -- (H) -- cycle;
    \draw [line width=0.25mm, black] (B) -- (F);
    \draw [line width=0.25mm, black] (C) -- (G);
    \draw [line width=0.25mm, black] (D) -- (H);

    \shade[ball color=blue!50, opacity=0.875] (0.5, 0.875, 0.5) circle (0.04166666666);
    \node[scale=0.24]  at (1.6, 0.875, 0.5) {Initial bubble};
    \draw [-Stealth, line width=0.25mm, dotted, blue!50!black!80] (1.25, 0.875, 0.5) -- (0.6, 0.875, 0.5);
    \node[scale=0.24]  at (-0.7, 0.9, 0.5) {Gravity};
    \draw [-Stealth, line width=0.25mm, dashed, green!50!black!80] (-0.5, 1.2, 0.5) -- (-0.5, 0.6, 0.5);
    
\end{tikzpicture}
    \caption{Setup}
    \label{fig:sketch 3D rising bubble problem: setup}
\end{subfigure}
\begin{subfigure}{0.32\textwidth}
\includegraphics[scale = 0.3]{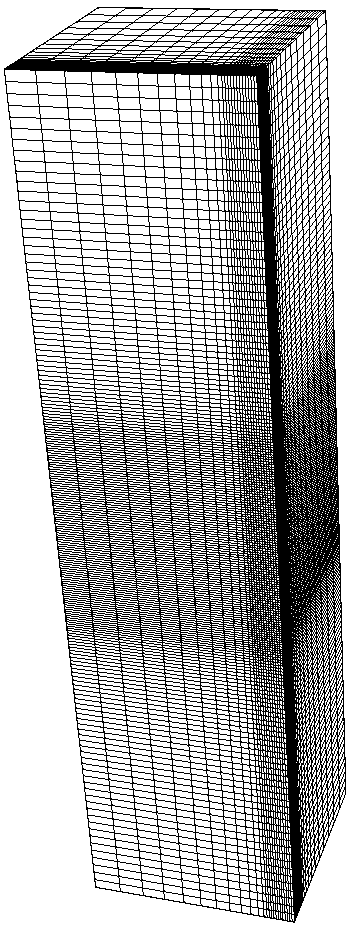}
    \caption{Coarsest mesh, quarter of domain}
    \label{fig:sketch 3D rising bubble problem: mesh}
\end{subfigure}
    \caption{Three-dimensional rising bubble problem.}
    \label{fig:sketch 3D rising bubble problem}
\end{figure}
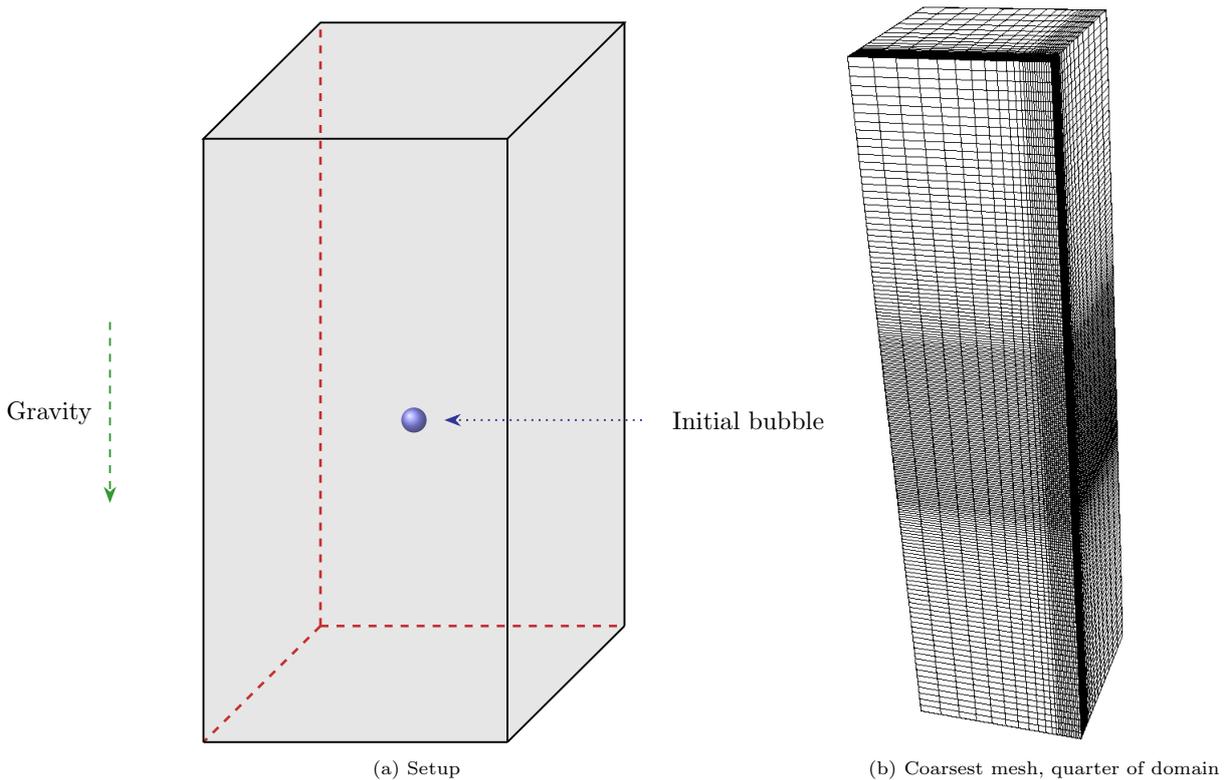

We note that the problem is symmetric in the planes $x=6$ and $z=6$. In order to reduce the computational effort, we only simulate the quarter $[0, 6] \times [0,24] \times [0,6]$ of the domain and apply symmetry boundary conditions.
To accurately capture the bubble dynamics, we use a stretched single patch mesh with a uniform mesh size $h = {\rm min}_K h_K$ inside the region of the bubble $[5.2, 6] \times [9, 15] \times [5.2, 6]$. The element size gradually increases towards the boundary of the computational domain, see \cref{fig:sketch 3D rising bubble problem: mesh}. 

The system is now characterized by $5$ dimensionless quantities: $\mathbb{A}{\rm r}$, $\mathbb{E}{\rm o}$, $\mathbb{C}{\rm n}$, $\rho_1/\rho_2$ and $\nu_1/\nu_2$.  We consider three different cases, for which the dimensionless parameters are given in \cref{table: parameters 3D RB cases}.

\begin{table}[ht]
\centering
\begin{tabularx}{\textwidth}{XXXXXXXXX}
Case & $\rho_1/\rho_2$ & $\nu_1/\nu_2$ & $\mathbb{A}{\rm r}$ & $\mathbb{E}{\rm o}$ \\[4pt]
\hline\\[-6pt]
\hspace{0.25cm}1 & $1000$ & $100$ & $1.671$ & $17.7$ \\[6pt]
\hspace{0.25cm}2 & $1000$ & $100$ & $15.24$ & $243$   \\[6pt]
\hspace{0.25cm}3 & $1000$ & $100$ & $30.83$ & $339$   \\[6pt]
\hline
\end{tabularx}
\caption{Parameters for the three-dimensional rising bubble cases.}
\label{table: parameters 3D RB cases}
\end{table}

To preclude the influence of the unresolved flow features and domain boundaries, we compare computational results on different meshes. We compare the results of case 2 on a coarse mesh, a medium mesh, and a fine mesh with total number of elements:  $24 \times 192 \times 24 = 110592$,  $48 \times 384 \times 48 = 884736$, and  $96 \times 768 \times 96 = 7077888$, with $h = 1/15, 1/30, 1/60$ respectively. We select the Cahn number as $\mathbb{C}{\rm n}=0.72 h$ and the time step size as $\Delta t_n = 0.075 h$. In \cref{fig: Case B Rising bubble 3D contours} we visualize the final bubble shape for the different meshes. We observe that the influence of the mesh is small when comparing the medium and fine mesh results. \cref{fig: Case B Re} shows the Reynolds number of the bubble $(\rho_1 v_b D_0)/\nu_1$ for the three meshes. The results of the medium and fine mesh show negligible differences. In the following we work with the fine mesh discretization.

\begin{figure}[!ht]
\begin{subfigure}{0.33\textwidth}
\centering
\includegraphics[width=0.80\textwidth]{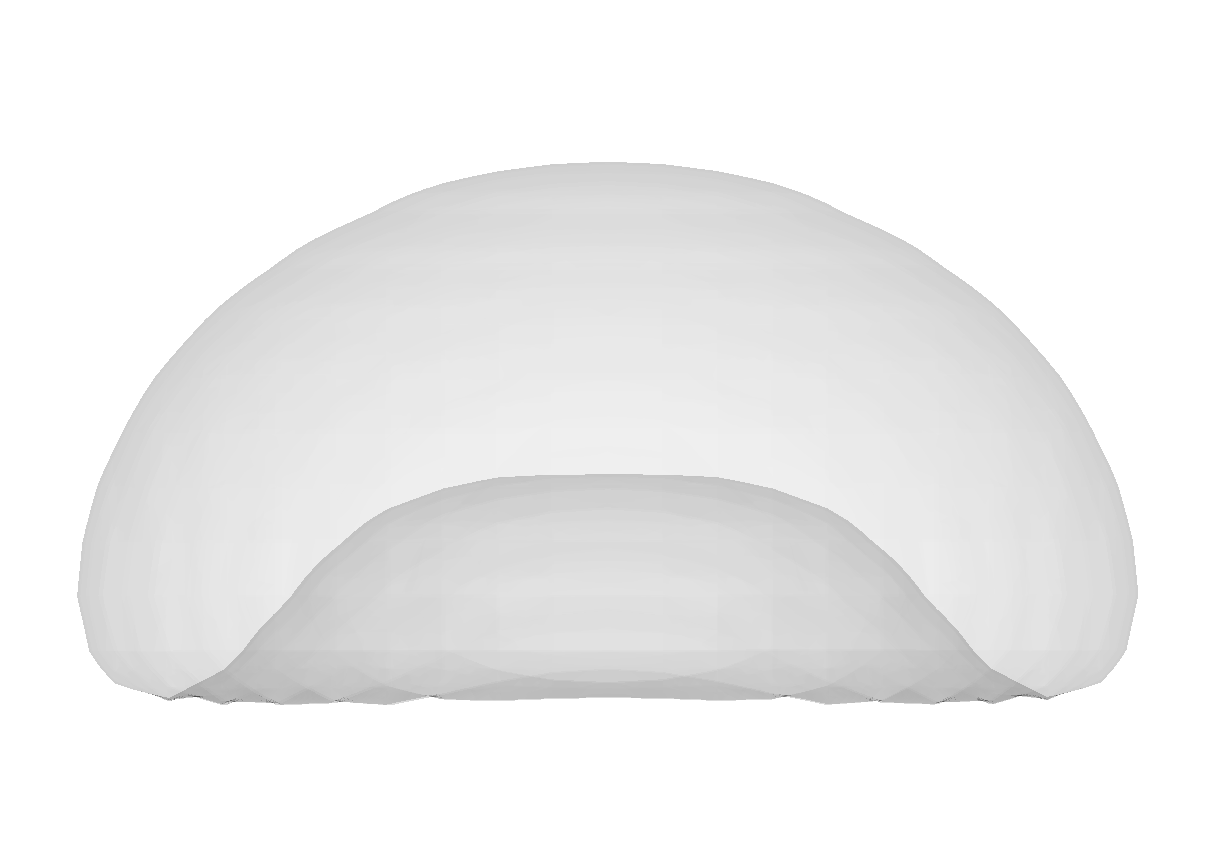}
\caption{Coarse mesh}
\end{subfigure}
\begin{subfigure}{0.33\textwidth}
\centering
\includegraphics[width=0.80\textwidth]{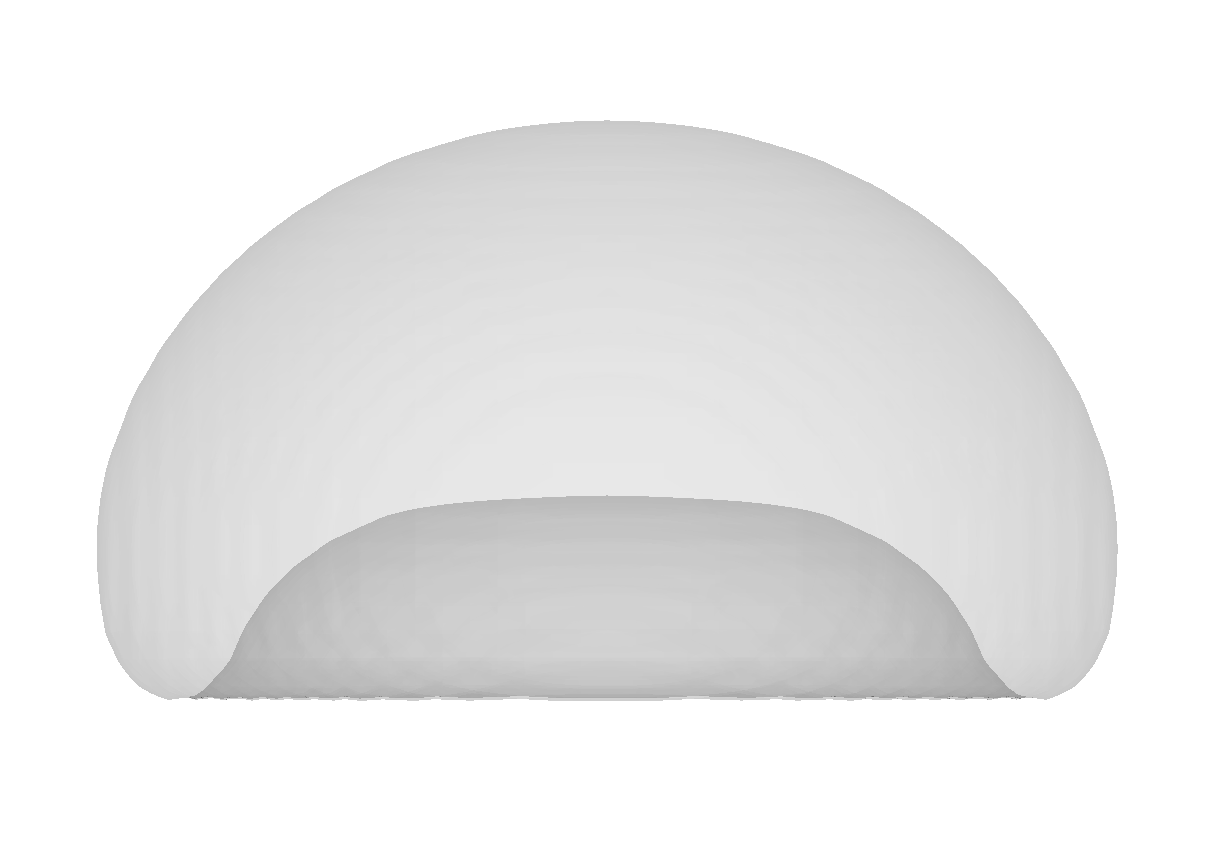}
\caption{Medium mesh}
\end{subfigure}
\begin{subfigure}{0.33\textwidth}
\centering
\includegraphics[width=0.80\textwidth]{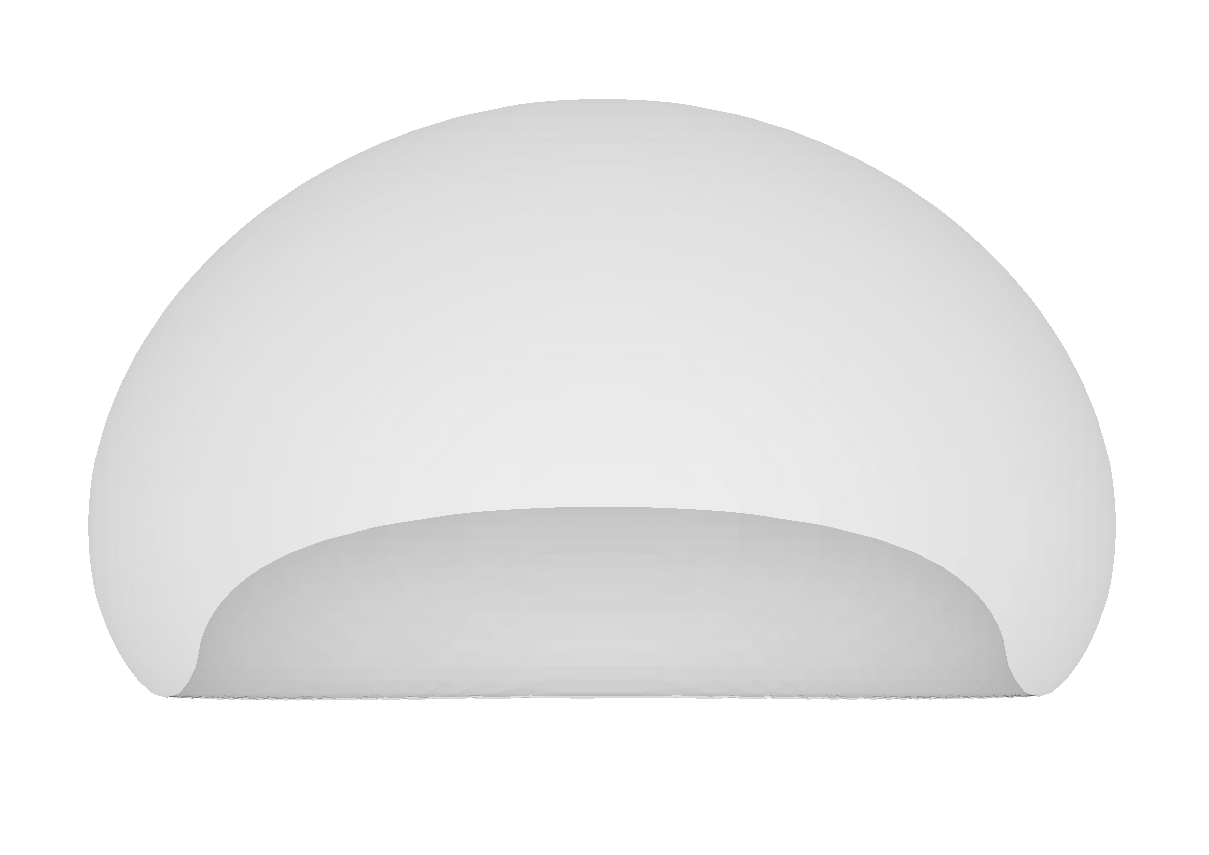}
\caption{Fine mesh}
\end{subfigure}
\caption{Rising bubble problem. Case 2. Final bubble shape for different mesh sizes.}
\label{fig: Case B Rising bubble 3D contours}
\end{figure}

\begin{figure}[!h]
\centering
\includegraphics[width=0.55\textwidth]{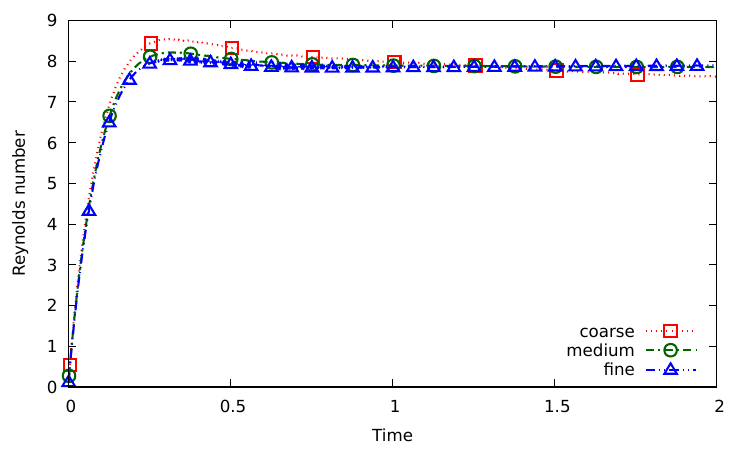}
\caption{Case 2. Reynolds number for different meshes.}
\label{fig: Case B Re}
\end{figure}

\begin{table}[!h]
\centering
\begin{tabularx}{\textwidth}{XXXXXXXXX}
Case                & Bhaga and Weber &  Yan et al. & Hua et al.  & Amaya-Bower and Lee & Current\\[4pt]
\hline\\[-6pt]
\hspace{0.25cm}1    & $0.232$         & $0.2593$    & $0.182$     &                     & $0.211$     \\[6pt]
\hspace{0.25cm}2    & $7.77$          & $7.5862$    & $7.605$     & $6.2$               & $7.90$     \\[6pt]
\hspace{0.25cm}3    & $18.3$          & $18.1717$   & $17.758$    & $15.2$              & $17.5$    \\[6pt]
\hline
\end{tabularx}
\caption{Final Reynolds number for the three-dimensional rising bubble cases. Left to right: Bhaga and Weber \cite{bhaga1981bubbles}, Yan et al. \cite{yan2019isogeometric}, Hua et al. \cite{hua2008numerical}, Amaya-Bower and Lee \cite{amaya2010single} and the current NSCH computation.}
\label{table: final Re}
\end{table}
\newpage
In \cref{table: final Re} we show the final Reynolds numbers of each of the three cases, and visualize the final bubble shape in \cref{fig: rising bubble 3D contours}. The table and figure include experimental data of Bhaga and Weber \cite{bhaga1981bubbles}, and computational data Yan et al. \cite{yan2019isogeometric}, Hua et al. \cite{hua2008numerical} and Amaya-Bower and Lee \cite{amaya2010single}. We observe that the bubble largely remains spherical in case 1, while in cases 2 and 3 a significant deformation is visible. 
In cases 1 and 2, we observe that the terminal Reynolds numbers computed by our methodology show the closest match with the experimental data. For case 3, all computational results show lower terminal Reynolds numbers than the experimental results. Concerning the final bubble shape, we see that bubble shape of our computation matches well with that of Yan et al. \cite{yan2019isogeometric} and the experimental data \cite{bhaga1981bubbles}. Finally, to highlight the deformation differences among the three cases, we visualize the three-dimensional terminal bubble shapes in \cref{fig: rising bubble 3D vol}.




\begin{figure}[!h]
\begin{subfigure}{0.32\textwidth}
\centering
\includegraphics[width=0.95\textwidth]{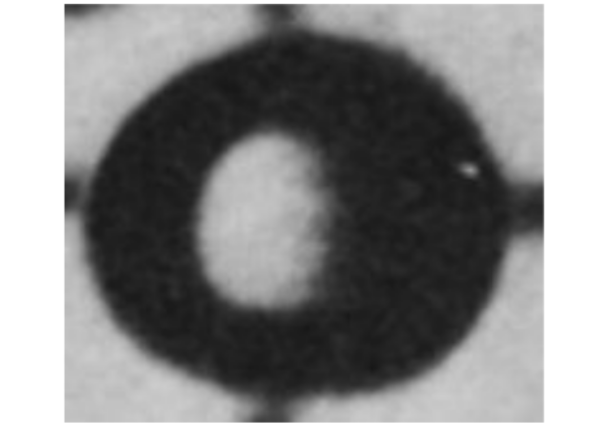}
\end{subfigure}
\begin{subfigure}{0.32\textwidth}
\centering
\includegraphics[width=0.95\textwidth]{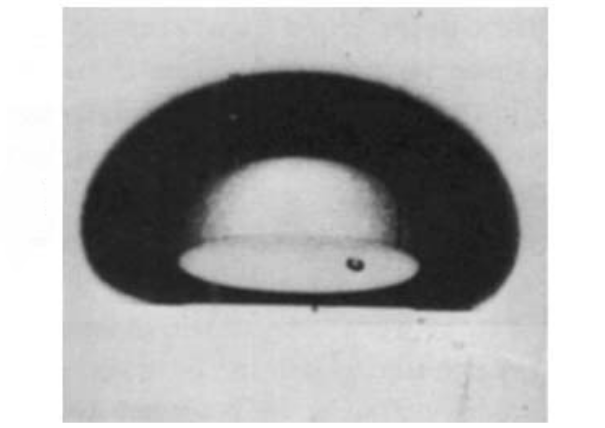}
\end{subfigure}
\begin{subfigure}{0.32\textwidth}
\centering
\includegraphics[width=0.95\textwidth]{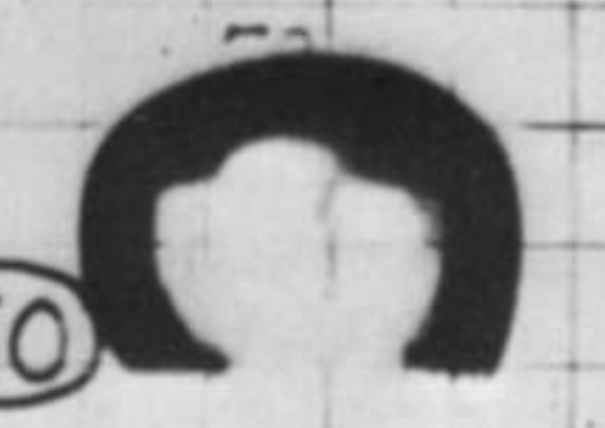}
\end{subfigure}
\begin{subfigure}{0.32\textwidth}
\centering
\includegraphics[width=0.95\textwidth]{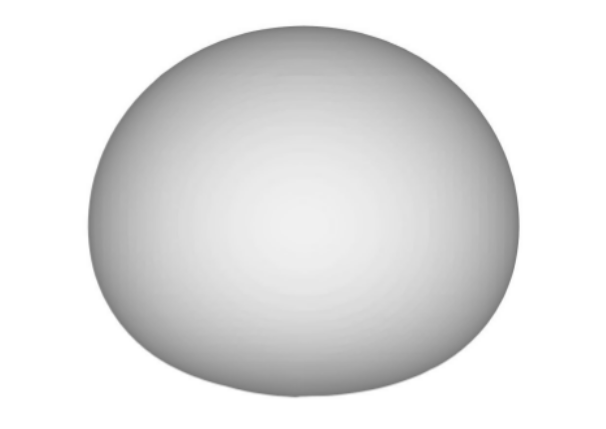}
\end{subfigure}
\begin{subfigure}{0.32\textwidth}
\centering
\includegraphics[width=0.95\textwidth]{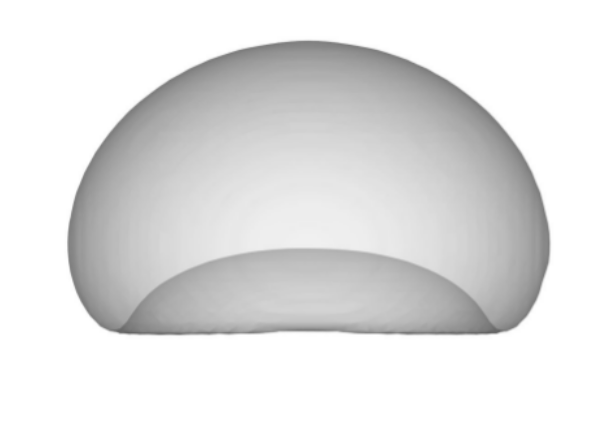}
\end{subfigure}
\begin{subfigure}{0.32\textwidth}
\centering
\includegraphics[width=0.95\textwidth]{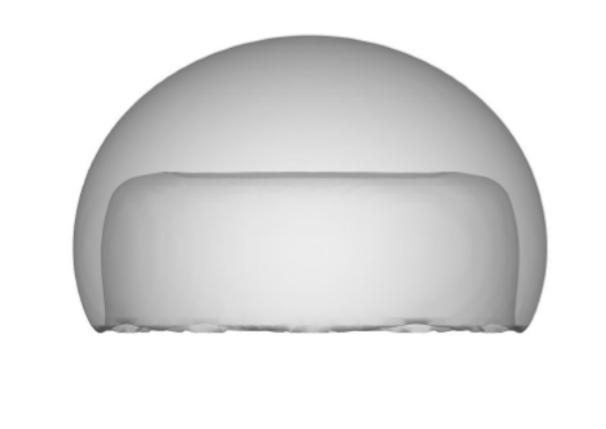}
\end{subfigure}
\begin{subfigure}{0.32\textwidth}
\centering
\includegraphics[width=0.95\textwidth]{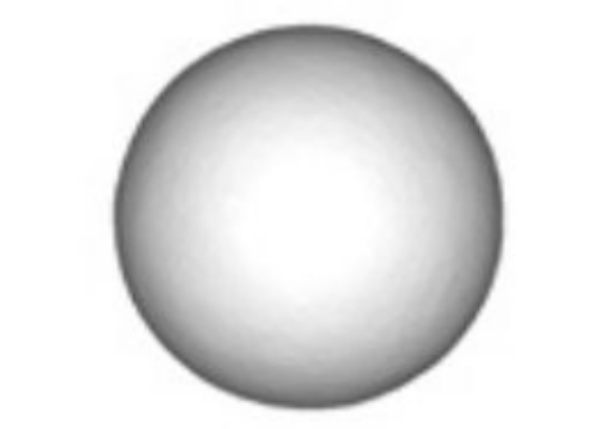}
\end{subfigure}
\begin{subfigure}{0.32\textwidth}
\centering
\includegraphics[width=0.95\textwidth]{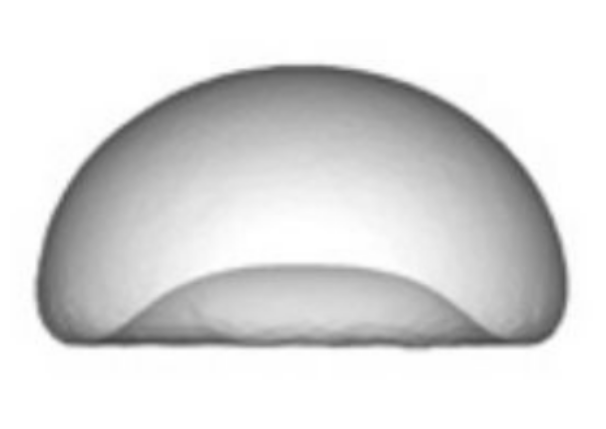}
\end{subfigure}
\begin{subfigure}{0.32\textwidth}
\centering
\includegraphics[width=0.95\textwidth]{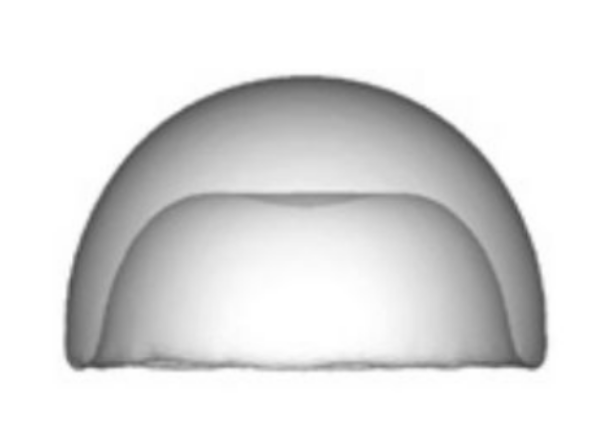}
\end{subfigure}
\begin{subfigure}{0.32\textwidth}
\centering
\includegraphics[width=0.95\textwidth]{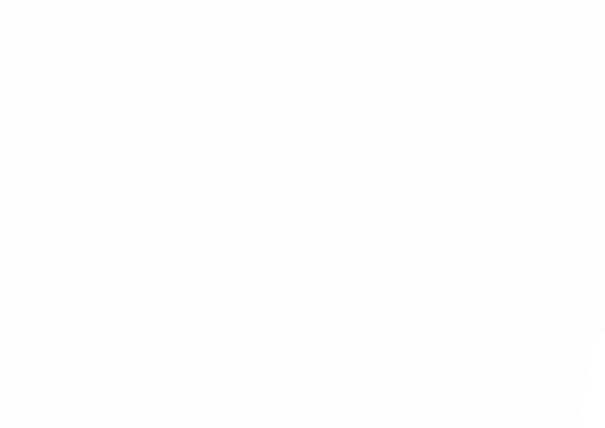}
\end{subfigure}
\begin{subfigure}{0.32\textwidth}
\centering
\includegraphics[width=0.95\textwidth]{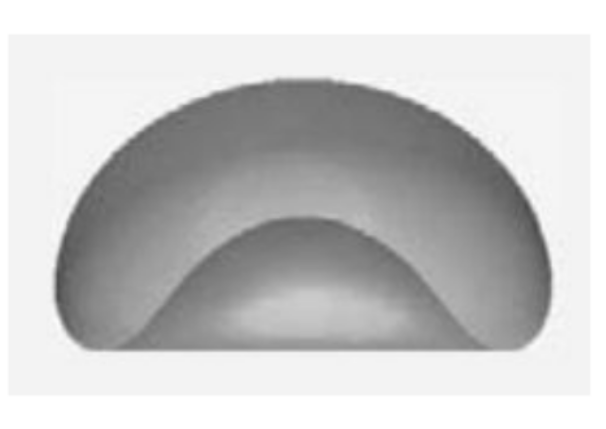}
\end{subfigure}
\begin{subfigure}{0.32\textwidth}
\centering
\includegraphics[width=0.95\textwidth]{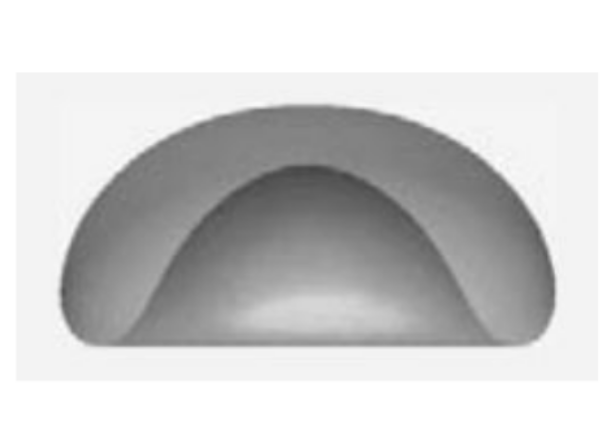}
\end{subfigure}
\begin{subfigure}{0.33\textwidth}
\centering
\includegraphics[width=0.80\textwidth]{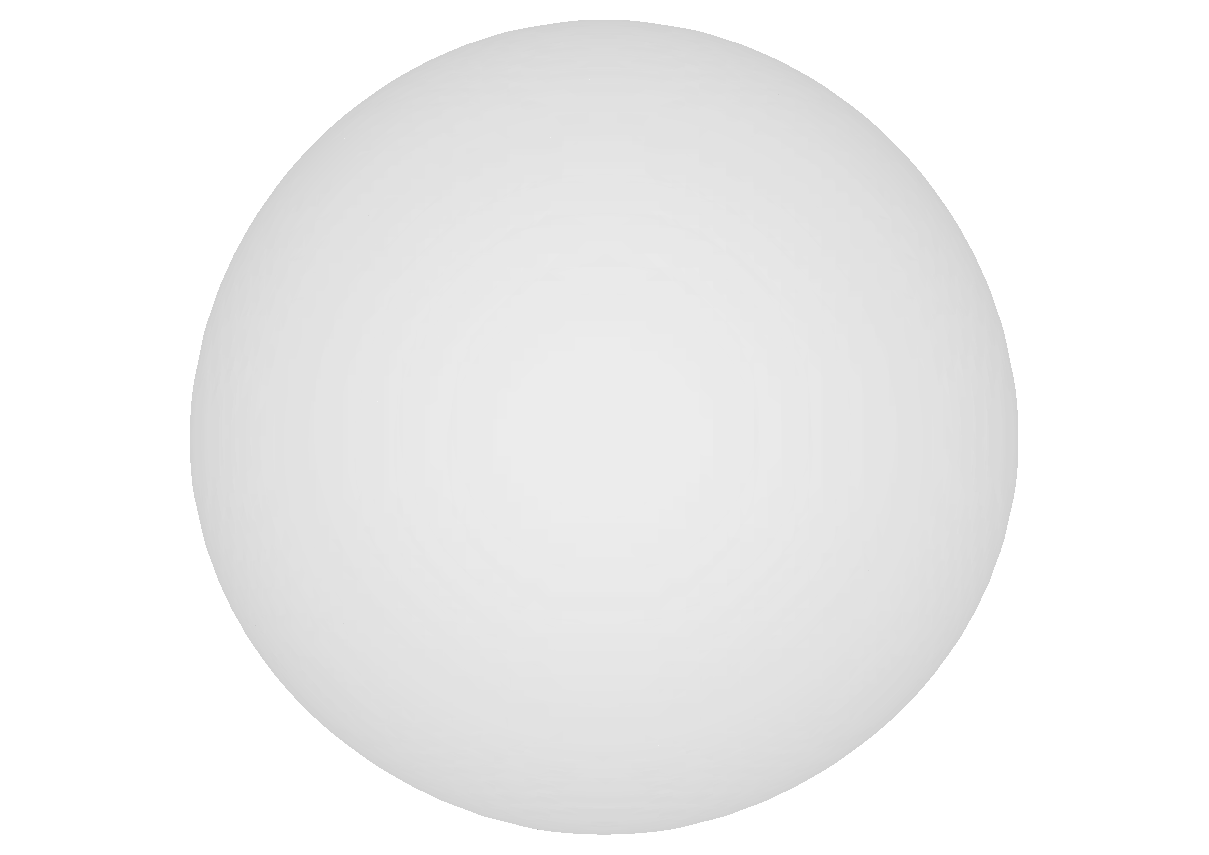}
\caption{Case 1}
\end{subfigure}
\begin{subfigure}{0.33\textwidth}
\centering
\includegraphics[width=0.80\textwidth]{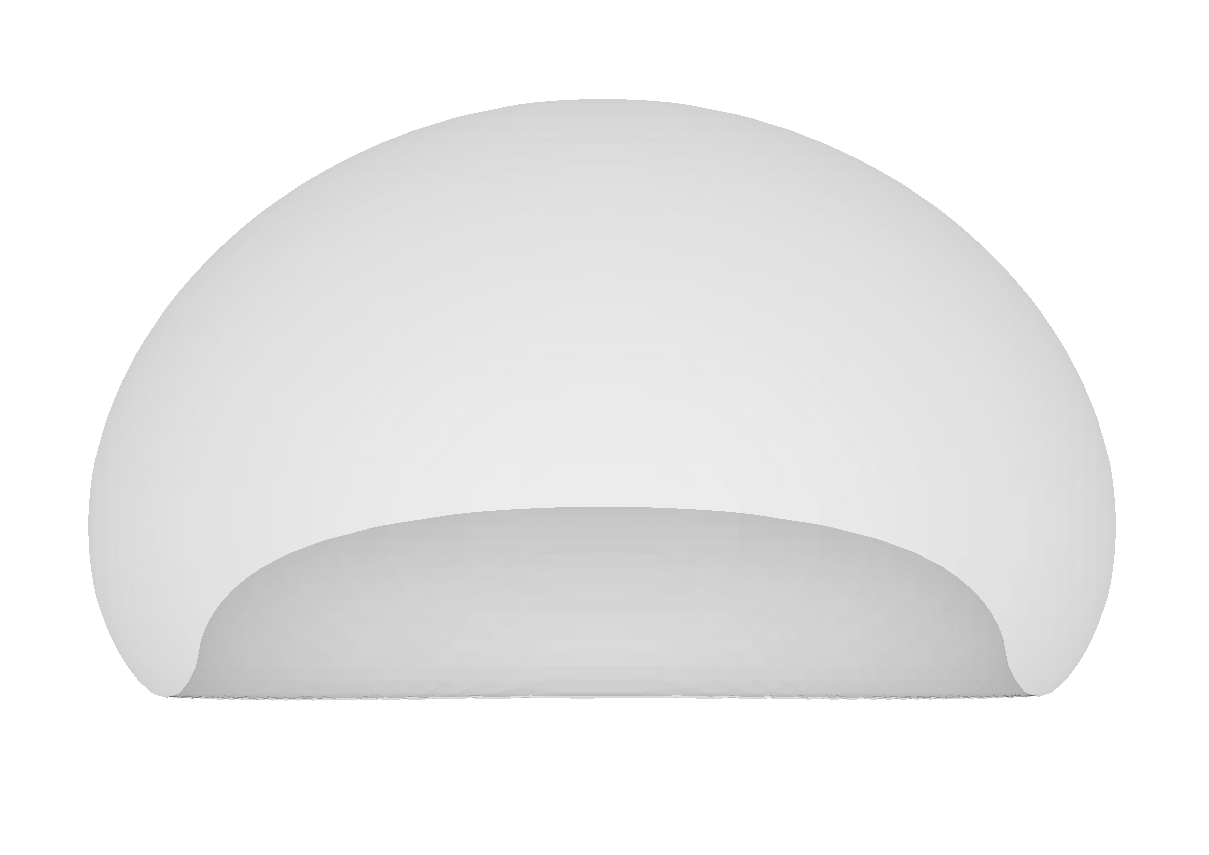}
\caption{Case 2}
\end{subfigure}
\begin{subfigure}{0.33\textwidth}
\centering
\includegraphics[width=0.80\textwidth]{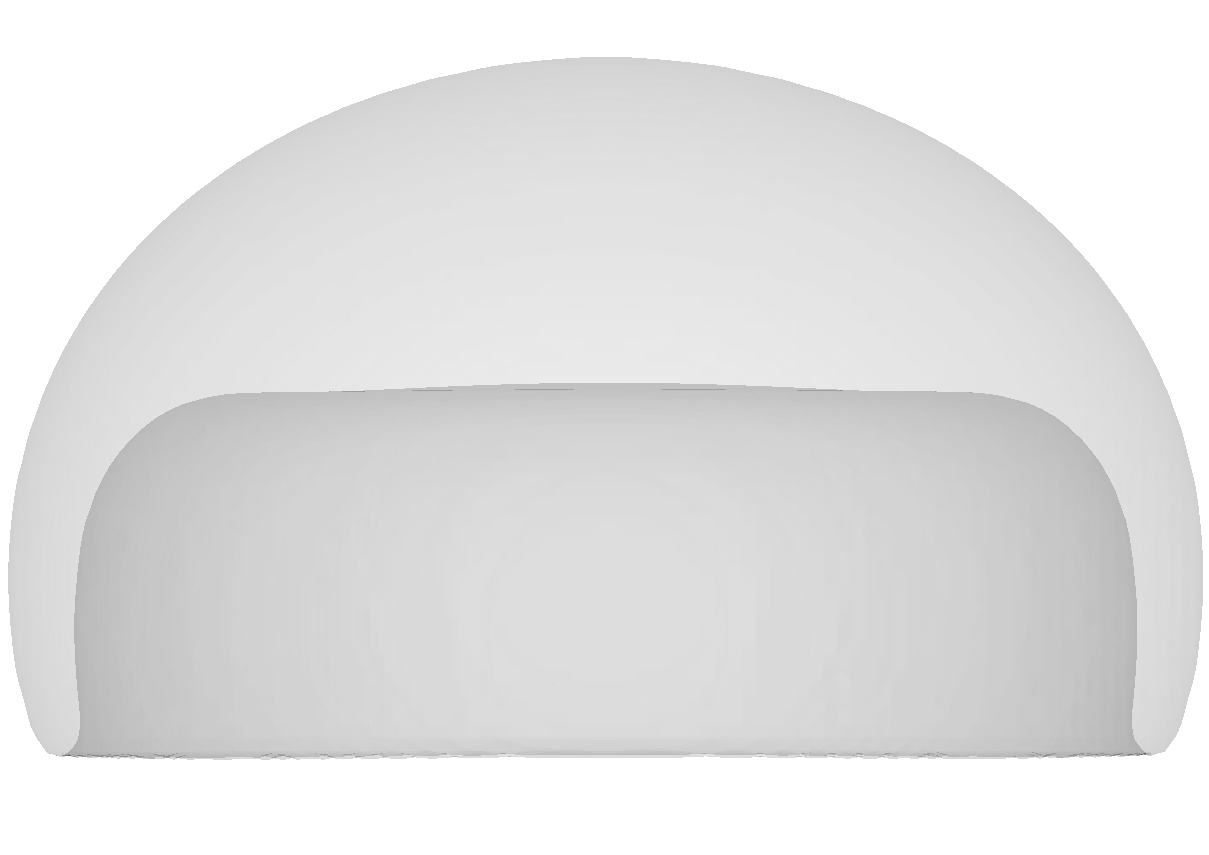}
\caption{Case 3}
\end{subfigure}
\caption{Rising bubble problem. Final bubble shape. Top to bottom: Bhaga and Weber \cite{bhaga1981bubbles}, Yan et al. \cite{yan2019isogeometric}, Hua et al. \cite{hua2008numerical}, Amaya-Bower and Lee \cite{amaya2010single}, and the current NSCH computation.}
\label{fig: rising bubble 3D contours}
\end{figure}

\clearpage

\begin{figure}[!h]
\begin{subfigure}{0.33\textwidth}
\centering
\includegraphics[width=0.80\textwidth]{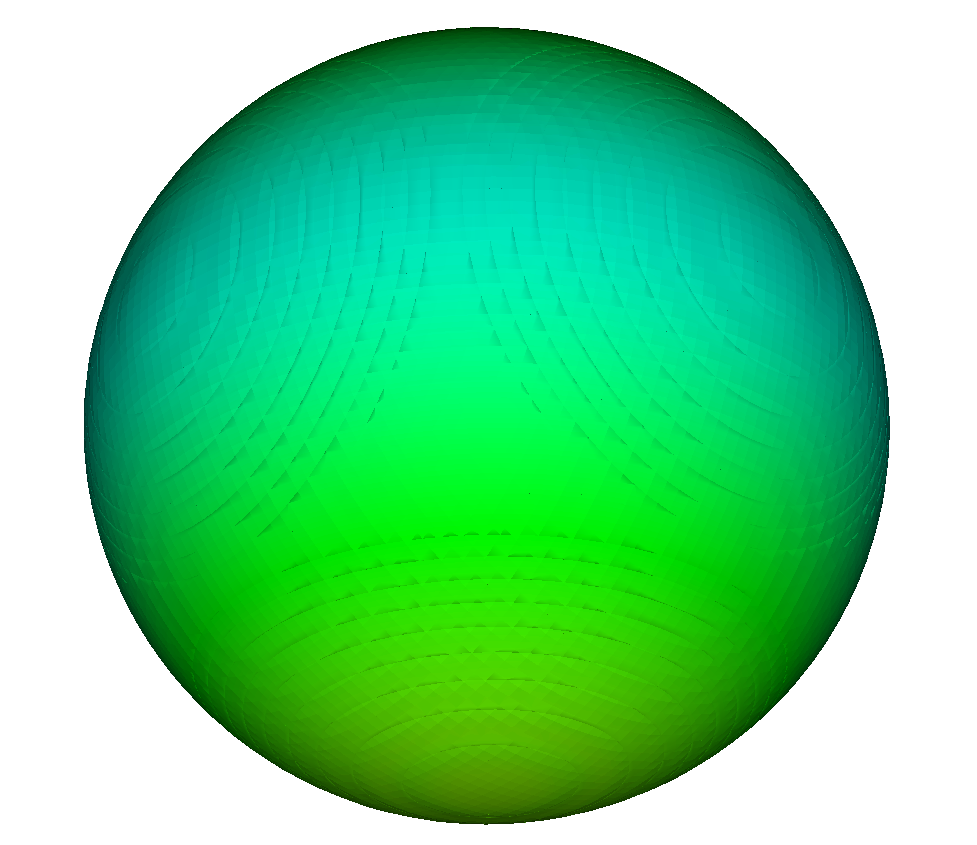}
\caption{Case 1}
\end{subfigure}
\begin{subfigure}{0.33\textwidth}
\centering
\includegraphics[width=0.80\textwidth]{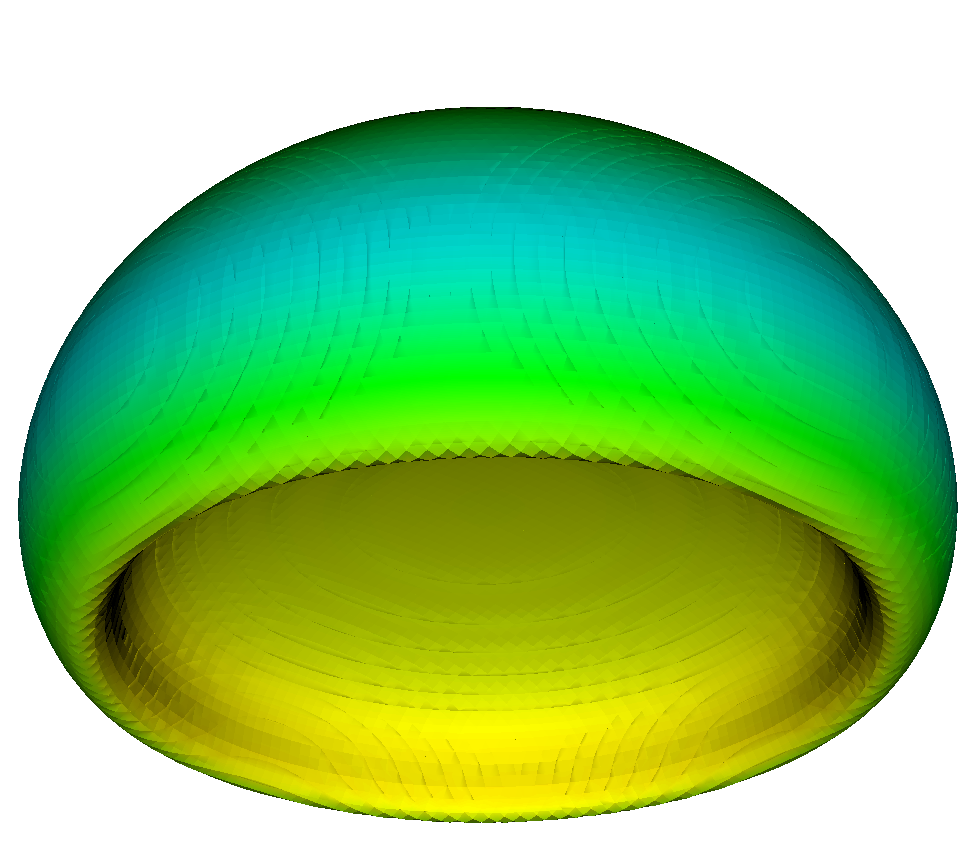}
\caption{Case 2}
\end{subfigure}
\begin{subfigure}{0.33\textwidth}
\centering
\includegraphics[width=0.80\textwidth]{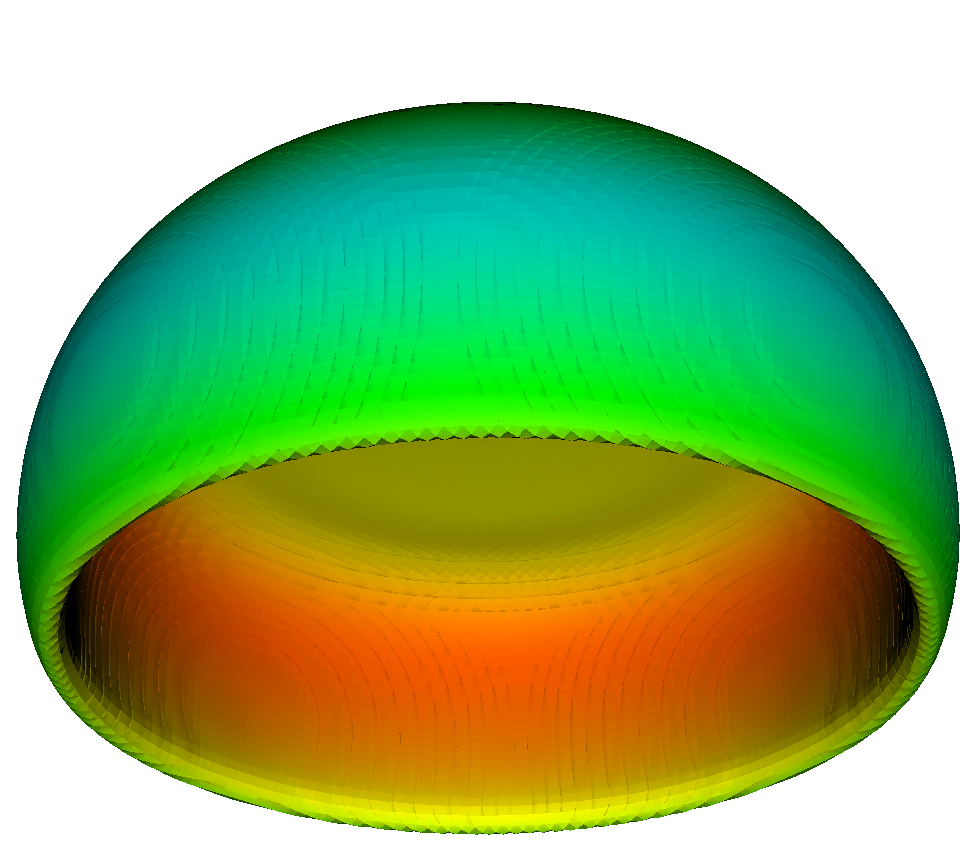}
\caption{Case 3}
\end{subfigure}
\caption{Rising bubble problem. Three-dimensional terminal bubble shape of the NSCH computations.}
\label{fig: rising bubble 3D vol}
\end{figure}

\subsection{Three-dimensional liquid filament contraction}

The dynamics of a contracting liquid filament is a fundamental problem that occurs in a wide range of problems in fluid dynamics, e.g. microfluidics, biological systems, spray, and inkjet printing. In this qualitative application case of a contracting liquid filament we validate our computational results against experimental data of Castrejon et al. \cite{castrejon2012breakup}. The fate of a contracting cylindrical liquid filament depends on the Ohnesorge number $\mathbb{O}{\rm h}$ (the relative importance of viscosity and surface tension) and the aspect ratio $\Gamma_0$. The aspect ratio is defined as $\Gamma_0 = L_0/D_0$ where $L_0$ is the initial length of the filament, and $D_0 = 2 R_0$ the width, and $R_0$ the radius.

We use a rectangular physical domain rectangular domain of size $[0,24]\times[0,180]\times[0,24]$ in which the liquid filament is placed at the center.
The initial phase field profile is:
\begin{align}\label{eq: init phi 3D LC}
  \phi^h_0(\mathbf{x}) =
  \begin{cases}
    \tanh{\dfrac{R_0-\sqrt{(x-12)^2+(y-32)^2 + (z-12)^2}}{\mathbb{C}{\rm n}\sqrt{2}}} & \text{if}~y < 32,\\[8pt]
    \tanh{\dfrac{R_0-\sqrt{(x-12)^2 + (z-12)^2}}{\mathbb{C}{\rm n}\sqrt{2}}} & \text{if}~32 < y < 148,\\[8pt]
    \tanh{\dfrac{R_0-\sqrt{(x-12)^2+(y-148)^2 + (z-12)^2}}{\mathbb{C}{\rm n}\sqrt{2}}} & \text{if}~148 < y.
  \end{cases}
\end{align}
We applied at all boundaries the no-penetration boundary condition. \cref{fig: init filament} shows the sketch of the problem setup.

\begin{figure}[h]
    \centering
\begin{tikzpicture}[scale=4, transform shape]
    \coordinate (A) at (0,0,0);
    \coordinate (B) at (1,0,0);
    \coordinate (C) at (1,0,1);
    \coordinate (D) at (0,0,1);
    \coordinate (E) at (0,2,0);
    \coordinate (F) at (1,2,0);
    \coordinate (G) at (1,2,1);
    \coordinate (H) at (0,2,1);

    \fill[gray!20] (A) -- (B) -- (C) -- (D) -- cycle;
    \fill[gray!20] (E) -- (F) -- (G) -- (H) -- cycle;
    \fill[gray!20] (B) -- (F) -- (G) -- (C) -- cycle;
    \fill[gray!20] (A) -- (E) -- (H) -- (D) -- cycle;
    \fill[gray!20] (A) -- (B) -- (F) -- (E) -- cycle;
    
    \draw [line width=0.35mm, red!70!black!80, dashed] (A) -- (B) ;
    \draw [line width=0.35mm, red!70!black!80, dashed] (A) -- (D) ;
    \draw [line width=0.35mm, red!70!black!80, dashed] (A) -- (E);
    \draw [line width=0.25mm, black] (D) -- (C) ;
    \draw [line width=0.25mm, black] (C) -- (B) ;
    \draw [line width=0.25mm, black] (E) -- (F) -- (G) -- (H) -- cycle;
    \draw [line width=0.25mm, black] (B) -- (F);
    \draw [line width=0.25mm, black] (C) -- (G);
    \draw [line width=0.25mm, black] (D) -- (H);


    \foreach \y in {0.6,0.605,...,1.4}{
        \shade[ball color=blue!50, opacity=0.5] (0.5, \y, 0.5) circle (0.04166666666);
    }
    \node[scale=0.24]  at (1.8, 0.875, 0.5) {Initial liquid filament};
    \node[scale=0.24]  at (1.615, 0.775, 0.5) {Length: $L_0$};
    \node[scale=0.24]  at (1.613, 0.675, 0.5) {Width: $D_0$};
    \node[scale=0.24]  at (1.89, 0.575, 0.5) {Aspect ratio: $\Gamma_0=L_0/D_0$};
    \draw [-Stealth, line width=0.25mm, dotted, blue!50!black!80] (1.25, 0.875, 0.5) -- (0.6, 0.875, 0.5);
    \node[scale=0.24]  at (-0.7, 0.9, 0.5) {Gravity};
    \draw [-Stealth, line width=0.25mm, dashed, green!50!black!80] (-0.5, 1.2, 0.5) -- (-0.5, 0.6, 0.5);
\end{tikzpicture}
    \caption{Setup of the liquid filament contraction problem.}
    \label{fig: init filament}
\end{figure}
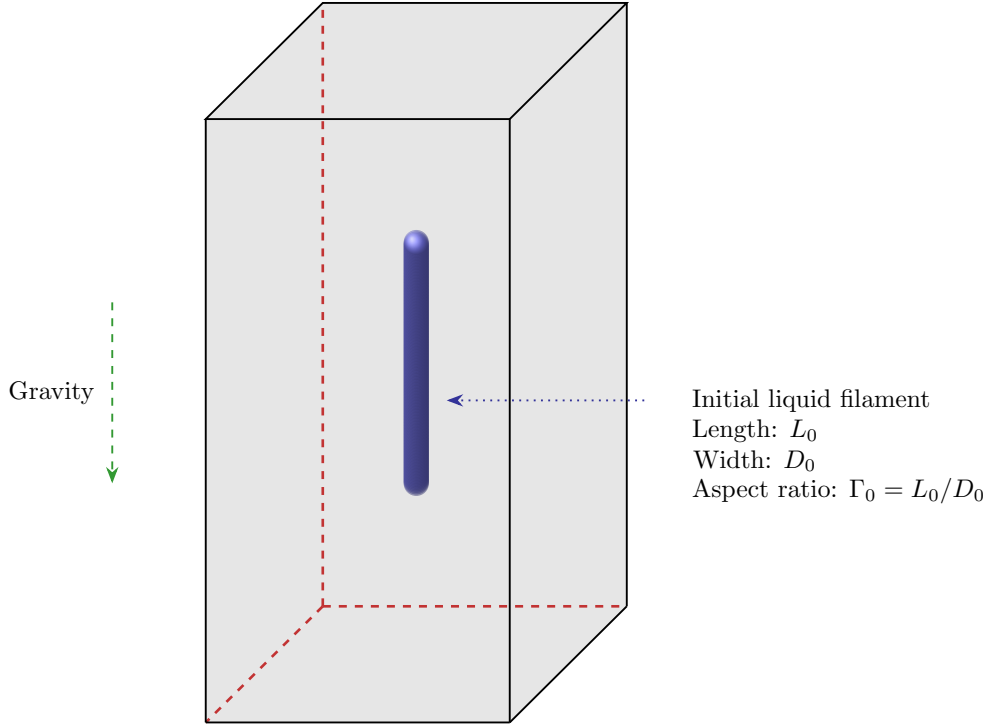

Similar as in the case of the rising bubble, the problem is symmetric in the planes $x=12$ and $z=12$. In order to reduce the computational effort, we only simulate the quarter $[0, 12] \times [0,180] \times [0,12]$ of the domain and apply symmetry boundary conditions.
Again, we use a stretched single patch mesh with a uniform mesh size $h = {\rm min}_K h_K$ inside the region of the ligament $[9, 12] \times [20, 160] \times [9, 12]$. We use a mesh with $16 \times 600 \times 16 = 153600$ elements with $h=0.25$.

A liquid filament either remains a single filament or breaks up into possibly multiple smaller filaments. Considering filaments that are initially at rest, those that have a large $\mathbb{O}{\rm h}$ and small $\Gamma_0$ typically remain a single body whereas filaments with small $\mathbb{O}{\rm h}$ and large $\Gamma_0$ tend to break up. It is difficult in an experimental setup to ensure that a filament is initially completely at rest. To minimize the influence of start-up effects, we study the contraction of a long liquid filament.

We consider a pure water liquid filament at rest at $22^{\circ}{\rm C}$ in the air.  The water viscosity is measured at $\nu_1 = 100\times 10^{-3} {\rm Pa~s}$, and the water density is $\rho_1 = 1000 \times 10^3{\rm g}/{\rm m}^3$ and the gravitational constant is $g = 9.81 {\rm m}/{\rm s}^2$. The dimensions are $R_0 = 0.13\times 10^{-3} {\rm m}$ and $\Gamma_0=59$. We select as reference values:
\begin{subequations}
    \begin{align}
        X_0 =&~ R_0, \\
        T_0 =&~ \sqrt{\dfrac{\rho_1 R_0^3}{\sigma}},\\
        U_0 =&~ \dfrac{R_0}{T_0},
    \end{align}
\end{subequations}
where $T_0$ is the capilary time scale. As a consequence the dimensionless numbers become:
\begin{subequations}
    \begin{align}
        \mathbb{R}{\rm e} =&~ \mathbb{O}{\rm h}^{-1},\\
        \mathbb{F}{\rm r} =&~ \mathbb{E}{\rm o}^{-1/2},\\
        \mathbb{W}{\rm e} =&~ 1,
    \end{align}
\end{subequations}
where the Ohnesorge number ($\mathbb{O}{\rm h}$) and the E\"{o}tv\"{o}s number ($\mathbb{E}{\rm o}$) are given by:
\begin{subequations}\label{eq: dimensionless quantities 3}
\begin{align}
     \mathbb{O}{\rm h} =&~ \dfrac{\nu_1}{\sqrt{\rho_1 R_0 \sigma}},\\
     \mathbb{E}{\rm o} =&~ \frac{\rho_1 g R_0^2}{\sigma}. 
\end{align}
\end{subequations}
The E\"{o}tv\"{o}s number ($\mathbb{E}{\rm o}$) describes the relative importance of gravity and surface tension, and is also known as the Bond number ($\mathbb{B}{\rm o}$). 

The system is now characterized by $6$ dimensionless quantities: $\Gamma_0, \mathbb{O}{\rm h}$, $\mathbb{E}{\rm o}$, $\mathbb{C}{\rm n}$, $\rho_1/\rho_2$ and $\nu_1/\nu_2$. The experiment is conducted over a time period of $10.7 \times 10^{-3} {\rm s}$, which corresponds to final time $T_{\rm end} = 62.6867~T_0$. We set $\Delta t_n = 0.002~ T_{\rm end}$. The Cahn number is taken as $\mathbb{C}{\rm n} = 0.64 h$.  The other dimensionless quantities are given in \cref{table: parameters 3D LC case}.

\begin{table}[ht]
\centering
\begin{tabularx}{\textwidth}{XXXXXXXXX}
\hspace{0.25cm}$\rho_1/\rho_2$ & $\nu_1/\nu_2$ & $\mathbb{O}{\rm h}$ & $\mathbb{E}{\rm o}$ \\[4pt]
\hline\\[-6pt]
 \hspace{0.25cm} $1000$ & $100$ & $1.01$ & $0.0022$ \\[6pt]
\hline
\end{tabularx}
\caption{Parameters for the three-dimensional ligament contraction case.}
\label{table: parameters 3D LC case}
\end{table}

In \cref{fig: Ligament 3D contours} we qualitatively compare our computational results of the contracting liquid filament with the experimental data of Castrejon et al. \cite{castrejon2012breakup}. We observe an overall good agreement between the computation and the experiment. In particular, the length of the liquid filament matches well. We see a slight deviation in the vertical position of the filament. This might be the consequence of a non-zero downwards velocity of the experiment at the initial time. Finally, we note that experimental results show that the filament does not break up (this is hard to detect in the last frame). In fact, Castrejon et al. \cite{castrejon2012breakup} find that at this $\mathbb{O}{\rm h}$ number the filament never breaks up, even for very large aspect ratios $\Gamma_0$. Our computational results confirm this observation.

\begin{figure}[!ht]
\begin{subfigure}{0.24\textwidth}
\centering
\includegraphics[width=0.8\textwidth]{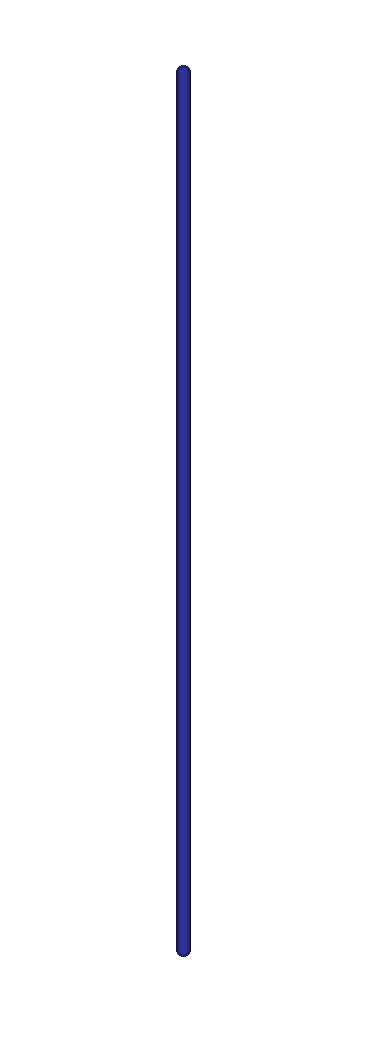}
\caption{$t=0.0\times 10^{-3} {\rm s}$ (num)}
\end{subfigure}
\begin{subfigure}{0.24\textwidth}
\centering
\includegraphics[width=0.80\textwidth]{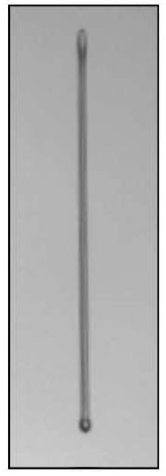}
\caption{$t=0.0\times 10^{-3} {\rm s}$ (exp)}
\end{subfigure}
\begin{subfigure}{0.24\textwidth}
\centering
\includegraphics[width=0.8\textwidth]{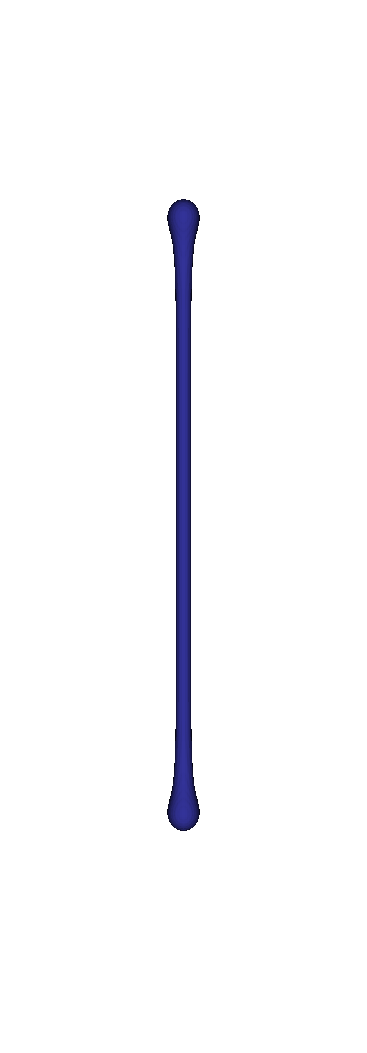}
\caption{$t=4.0\times 10^{-3} {\rm s}$ (num)}
\end{subfigure}
\begin{subfigure}{0.24\textwidth}
\centering
\includegraphics[width=0.80\textwidth]{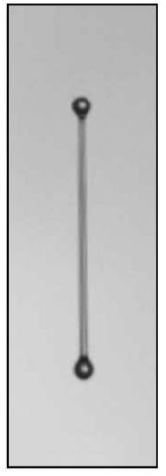}
\caption{$t=4.0\times 10^{-3} {\rm s}$ (exp)}
\end{subfigure}
\begin{subfigure}{0.24\textwidth}
\centering
\includegraphics[width=0.8\textwidth]{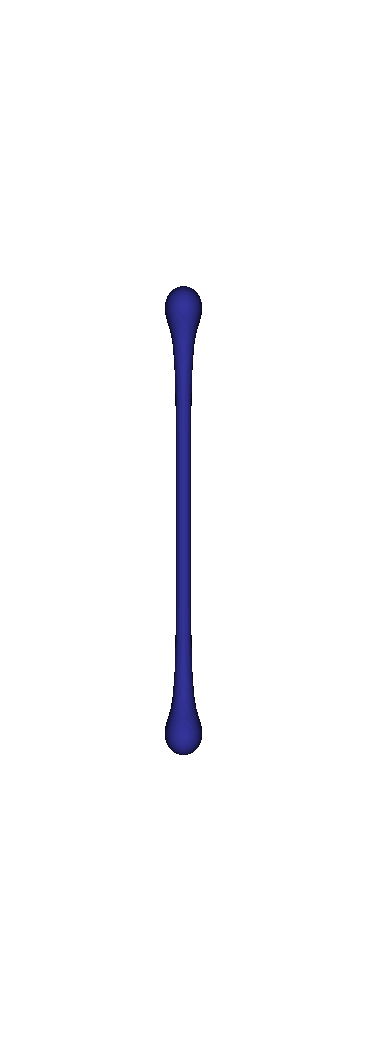}
\caption{$t=6.3\times 10^{-3} {\rm s}$ (num)}
\end{subfigure}
\begin{subfigure}{0.24\textwidth}
\centering
\includegraphics[width=0.80\textwidth]{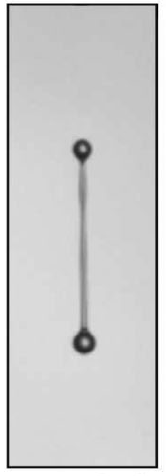}
\caption{$t=6.3\times 10^{-3} {\rm s}$ (exp)}
\end{subfigure}
\begin{subfigure}{0.24\textwidth}
\centering
\includegraphics[width=0.8\textwidth]{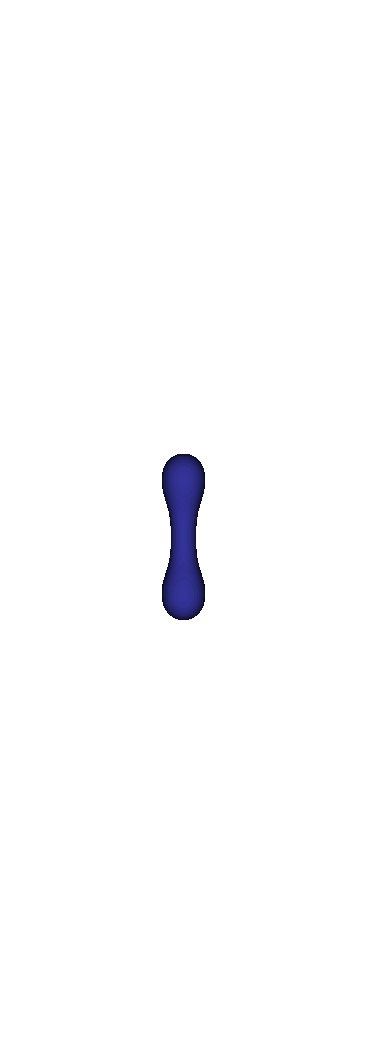}
\caption{$t=10.7\times 10^{-3} {\rm s}$ (num)}
\end{subfigure}
\begin{subfigure}{0.24\textwidth}
\centering
\includegraphics[width=0.80\textwidth]{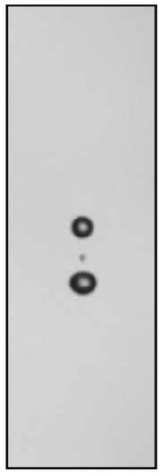}
\caption{$t=10.7\times 10^{-3} {\rm s}$ (exp)}
\end{subfigure}
\caption{Ligament contraction problem. Shape of the ligament using the numerical simulation (num) experimental results (exp).}
\label{fig: Ligament 3D contours}
\end{figure}

\clearpage
\section{Conclusion and outlook}\label{sec: summary outlook}

In this paper we have presented the consistent Navier-Stokes Cahn-Hilliard (NSCH) model with non-matching densities, alongside with a numerical methodology, and verification and validation studies. The NSCH model is a diffuse-interface energy-stable model that describes the motion of a mixture of fluids, and is derived from continuum mixture theory. The NSCH model is a single model, not a collection of models, that is invariant to the choice of fundamental variables. The mixture theory derivation naturally leads to a formulation in terms of the mass-averaged velocity. In order to circumvent the challenges of directly discretizing this NSCH formulation, we have adopted an equivalent formulation in terms of the divergence-free volume-averaged velocity. We have discretized this formulation using divergence-conforming isogeometric spaces. Finally, we have verified and validated the the methodology using two-dimensional test cases, and three-dimensional benchmark computations of rising bubbles and the contraction of a liquid filament. These three-dimensional computations compare well with experimental data.
Concluding, we have proposed a NSCH computational framework that:
\begin{itemize}
    \item uses the consistent NSCH model which naturally emerges from continuum mixture theory;
    \item employs a non-standard form of the NSCH model which properly accounts for the diffuse flux contribution in the momentum equation;
    \item uses a degenerate mobility;
    \item reduces in the single-fluid regime to the incompressible Navier-Stokes equations;
    \item uses a pointwise divergence-free velocity discretization in the entire computational domain;
    \item permits large density ratios;
    \item has a stable interface, and thus refrains from interface stabilization techniques;
    \item shows good agreement with experimental data.
\end{itemize}

We delineate some potential avenues for future research. First, regarding the mathematical analysis, establishing sharp interface asymptotics remains an important open problem. This sharp interface problem is invariant to the choice of fundamental variables in the NSCH model. Second, in regards to computation, the extension to high Reynolds number flow necessitates widely applicable and robust turbulence models. We anticipate that (energy-dissipative) variational multiscale stabilization mechanisms are well-suited for this purpose, e.g. by building onto \cite{EiAk17ii,evans2020variational}. 

\section*{Acknowledgments}
The authors wish to thank Sebastian Aland for sharing the two-dimensional computational data of the Navier-Stokes Cahn-Hilliard models.
MtE was partly supported by the German Research Foundation (Deutsche Forschungsgemeinschaft DFG) via the Walter Benjamin project EI 1210/1-1. DS gratefully acknowledges support from the German Research Foundation (Deutsche Forschungsgemeinschaft DFG) via the Emmy Noether Award SCH 1249/2-1, and  project SCH 1249/9-1. The authors gratefully acknowledge the computing time provided to them on the high-performance computer Lichtenberg at the NHR Centers NHR4CES at TU Darmstadt. This is funded by the Federal Ministry of Education and Research, and the state governments participating on the basis of the resolutions of the GWK for national high performance computing at universities.

\bibliographystyle{unsrt}
\bibliography{references}

\end{document}